\documentclass[prd,reprint, nofootinbib,amsmath,amssymb, aps, floatfix]{revtex4-1}
\usepackage{dcolumn}

\usepackage{bm}
\usepackage{amssymb,amsmath,amsthm,graphicx,ulem}
\usepackage[linktoc=page]{hyperref}

\usepackage{graphicx,subfigure}
\usepackage{epsfig}
\usepackage{amsmath}
\usepackage{amsfonts}
\usepackage{amssymb}
\usepackage[usenames]{color}
\usepackage{enumerate}
\usepackage{bbold}

\usepackage[T1]{fontenc}
\newtheorem{prop}{Proposition}[section]
\newtheorem{thm}{Theorem}[section]

\usepackage{mathtools,xparse}

\numberwithin{equation}{section}

\begin{document}

\normalem

\theoremstyle{definition}
\newtheorem*{definition}{Definition}

\newtheorem{lemma}{Lemma}
\newtheorem{corollary}{Corollary}

\def\bra#1{\langle #1 |}
\def\ket#1{| #1 \rangle}
\def\inner#1#2{\langle #1 | #2 \rangle}

\def\ext{\mathrm{ext} \: }
\def\rep{\textrm{rep} }
\def\Mm{\textrm{Mm}}
\def\Loc{\mathrm{Loc}}
\def\EW{\mathrm{EW}}
\def\CW{\mathrm{CW}}

\def\be{\begin{equation}}
\def\ee{\end{equation}}
\def\bn{\begin{align}}
\def\ebn{\end{align}}

\title{The Boundary Dual of Bulk Local Operators}
\author{Fabio Sanches}
\email{fabios@berkeley.edu} 
\affiliation{Berkeley Center for Theoretical Physics and Department of Physics, University of California, Berkeley, CA 94720, USA }
\author{Sean J. Weinberg}%
 \email{sjasonw@physics.ucsb.edu}
\affiliation{Department of Physics, University of California, Santa Barbara, CA 93106, USA }%

\bibliographystyle{utcaps}

\begin{abstract}
We provide a procedure to determine if a given nonlocal operator in a large $N$ holographic CFT is dual to a local bulk operator on the geometry associated with a particular code subspace of the CFT. This procedure does not presuppose knowledge of the bulk geometry.  We are able to pick out local operators in a large region of the bulk, called the ``localizable region,'' that  can extend beyond event horizons in certain cases.  The method  relies heavily on the quantum-error correcting structure of AdS/CFT and, in particular, on entanglement wedge reconstruction.  As a byproduct of this machinery, we are able to reconstruct the metric in the localizable region up to a conformal factor. This suggests a connection between our program and the recent light-cone cut approach to bulk reconstruction.
\end{abstract}

\maketitle

\tableofcontents


\section{Introduction}
\label{intro}

The emergence of bulk locality in the AdS/CFT correspondence \cite{Maldacena:1997,Witten:1998} has yet to receive a satisfactory explanation in terms of the behavior of holographic CFTs. While gravity prohibits  exact locality in a quantum theory, when the gravitational coupling is sufficiently small, local physics must be  a good approximation in the bulk. There should be a manifestation of this ``emergence of locality'' in the boundary theory.


One way to tackle this issue is by studying the ways in which bulk degrees of freedom are encoded in the CFT. It is thus natural to ask if there is a boundary dual
 of local bulk fields in the regime where semiclassical field theory holds.  While the extrapolate dictionary \cite{BDHM} states that bulk fields at spacelike infinity 
 are dual to local operators on the boundary, points deep in the bulk require a nonlocal holographic description.
There are many well-known ways to reconstruct bulk fields in terms of nonlocal boundary operators \cite{BDHM, HKLL, KLL,Heemskerk:2012mn} with support in a variety of boundary regions.  
All of these procedures, however, require solving bulk equations of motion which presupposes knowledge of the bulk spacetime.  
If one were not explicitly told the metric in the bulk, is there any way to determine whether or not a given operator is bulk local?  To put this question  differently, 
is the concept of a local bulk operator in any way distinguished in the boundary theory?

The primary goal of this work is to address this question.  We will find that a powerful tool to this end is the concept of subregion duality.   The notion that a boundary domain of dependence should
be thought of as being dual to some region of the bulk,
which originally arose from considerations of causal wedge reconstruction, was made precise  recently by
\cite{Almheiri:2014lwa, JLMS, DHW}  where it was concluded that a bulk operator can be reconstructed in 
a subregion of the CFT if and only if its support is contained in the entanglement wedge of that CFT region \cite{DHW}.  This conclusion was made in the context of a new 
development in AdS/CFT: the role of quantum error correction.  It is now understood that a semiclassical bulk spacetime description is associated with a code subspace of the
boundary Hilbert space, and that various inequivalent boundary reconstructions of bulk operators become equivalent when restricted to the code subspace.

 This modern form of subregion duality will reveal a novel characterization of locality in the bulk. Given a holographic CFT and a code subspace dual to some unknown geometry, we will provide
 a procedure that can identify, up to certain caveats, whether or not an operator is dual to a bulk local operator.
 As a byproduct of our method, we are also able to reconstruct the causal structure (equivalently, the metric up to a conformal factor) of a large region in the bulk.  In some examples, 
  this region can penetrate event horizons.

\emph{Outline}.  We start, in section \ref{sec_2}, by reviewing the arguments and motivation for the quantum error correcting view in holography. In particular, we sketch the proof of \cite{DHW}
 that a bulk operator is reconstructable in a boundary region if and only if its support is contained entirely in the entanglement wedge of that boundary region.

Section \ref{sec_3} contains the major constructions of this work.
We define the notion of a \emph{superficially local operator} without making direct reference to the bulk.  These are bulk operators that are ``as local as the boundary can directly tell.''  Their defining
characteristic is the great variety of boundary regions in which they can be reconstructed. In a certain region of the bulk called the \emph{localizable region}, operators are local if and only if they are superficially local.
However, there are situations in which superficially local operators correspond to nonlocal bulk operators that are supported outside of the localizable region.  The bulk regions in which these problematic 
operators lie will be referred to as \emph{clumps}.  Fortunately, clumps appear to always be identifiable  from the boundary theory because they are associated with phase transitions.  Thus, they can be identified and thrown away,
leaving only the superficially local operators that are authentically dual to bulk local operators.

The set of superficially local operators can be given an equivalence relation  by identifying two operators when they can be reconstructed in exactly the same boundary regions.  After removing clumps, the set of
equivalence classes of superficially local operators is naturally identified with the bulk localizable region.  

In \ref{sec_4}, we note that the commutation relations amongst these operators reveals the causal structure in the localizable region.
Thus, we are able to reconstruct the metric in this portion of the bulk up to a conformal rescaling.  This approach is similar at heart to that of  \cite{Engelhardt:2016wgb} where a bulk reconstruction is accomplished
by means of light-cone cuts.  We argue, in fact, that there are numerous interesting connections between our approach and that involving cut singularities.

\section{Principles of Subregion Duality}
\label{sec_2}

This section provides a brief review of the quantum error correcting view of AdS/CFT.  Readers already familiar with the conclusions of \cite{Almheiri:2014lwa, DHW} may wish to proceed to section \ref{sec_3} 

\vspace{.5cm}
There is a zoo of different methods for expressing bulk fields in terms of CFT operators.  The extrapolate dictionary \cite{BDHM} gives a precise relationship between limiting values of bulk fields and CFT operators with corresponding scaling dimensions. It is also possible to express operators lying deeper in the bulk in therms of CFT quantities by solving equations of motion in the bulk \cite{HKLL,BDHM,KLL,Morrison,Heemskerk:2012mn}.  
Of these approaches, one of relevance for our considerations is the causal wedge reconstruction, which generalizes the Rindler reconstruction of \cite{HKLL}.  This prescription expresses local bulk fields in terms of CFT operators localized to a special boundary subregion. Specifically, if $R$ is region in the boundary with domain of dependence $D^\partial(R)$, and if $\CW(R) = {J}^+[D^\partial(R)] \cap {J}^-[D^\partial(R)]$ is the causal wedge \cite{Hubeny:2012wa} of $R$, then causal wedge reconstruction allows a bulk field in $\CW(R)$ to be expressed as a smeared operator in $D^\partial(R)$.\footnote{The smearing function has to be understood in a distributional sense. For details see \cite{Morrison,Bousso:2012mh}. Such subtleties will not be important for what follows. 
} 

Causal wedge reconstruction suggests the possibility that subregions in the boundary are enough to understand the physics of associated bulk subregions. However, despite what is suggested from the analysis of \cite{HKLL}, 
the causal wedge is, in general, not the largest possible region that a boundary subregion holographically describes in the semiclassical limit. Instead, the bulk region dual to a CFT region $R$ is the entanglement wedge of $R$, denoted
by $\EW(R)$ \cite{Wall:2012uf,Headrick:2014cta,Jafferis:2014lza}.  $\EW(R)$ can be defined as follows.  Let $\Sigma$ be a spacelike bulk surface that, after conformal compactification of $M$, is a Cauchy slice for the unphysical bulk spacetime.
Require that $\Sigma$ contains $R$ and its HRT surface $\ext R$. Let $S$ denote the part of $\Sigma$ between $R$ and $\ext R$.  The domain of dependence of $S$ (computed in the unphysical spacetime) is the entanglement
wedge of $R$.  It is known that $ \EW(R) \supseteq \CW(R)$ \cite{Headrick:2014cta}. As we review below, \cite{DHW} gave 
a precise sense in which a boundary region $R$ should really be thought of as being dual to its entanglement wedge.  This is the most refined and powerful known form of ``subregion duality'' \cite{Bousso:2012mh} in AdS/CFT.

Before discussing entanglement wedge reconstruction, we note that subregion duality, even in the form of \cite{HKLL}, raises major puzzles \cite{Almheiri:2014lwa}.  For example, an operator $\phi(p)$ deep within the bulk can be taken to lie in many different causal wedges. Thus, a causal wedge reconstruction of the form 
\be
\phi(p) = \int_{D^\partial(R)} K(p,x) O(x) dx
\ee
manifestly commutes with all operators in the complement region $\bar{R}$. This argument can be repeated for many different boundary regions and used to show that a bulk field $\phi(p)$ near the center of AdS can be written in a way that manifestly commutes with any given operator in the boundary.  This directly implies what should have been obvious: that each choice of reconstruction for $\phi(p)$ is a different operator in the CFT.  This is not an inconsistency.  
Various reconstructions of $\phi(p)$ are distinct CFT operators, but the CFT Hilbert space is much larger\footnote{The basic concept that semiclassical excitations give rise to exponentially small subspaces of a Hilbert space
describing quantum gravitational physics has played a role in many related areas.  See, e.g.,  \cite{Nomura:2013lia,Nomura:2014woa,Papadodimas:2015jra}} than the Hilbert space relevant for a bulk operator on a spacetime background.  The explanation of the multitude of distinct
CFT operators is therefore that there is a special subspace of the Hilbert space, the code subspace, which describes the states that $\phi(p)$ is defined on. The restriction of all reconstructions of $\phi(p)$ to this subspace reproduce $\phi(p)$.
This is a quantum-error correcting property of  the CFT: the action of different operators defined in different regions is the same when restricting to special subspaces called \emph{code subspaces}.

The necessity for such a redundant descriptions of bulk operators was made particularly obvious with the following argument \cite{Almheiri:2014lwa}  illustrated in figure \ref{fig_1} .
 Consider a partition the boundary into 3 equal regions $R_1, R_2, \rm{and }\ R_3$ which only have points on their boundaries in common. Taking  the vacuum state for simplicity, their causal wedges will not contain points that are close to the center of the bulk spacetime. Thus, there is no HKLL smearing over any one region that reconstructs a local bulk operator near the center.  However, the causal wedge of the union of any two regions $\CW(R_i \cup R_j)$ does contain the bulk point of interest and the HKLL procedure can be used. The different choices cannot represent the same CFT operator, since their support is on causally disconnected regions.\footnote{
The mutual intersection actually includes points on the boundaries of the $R_i$. However, repeating the argument with slightly different regions circumvents the possibility that the reconstruction of $\phi$ is achieved only in the algebra of $\partial R_i$
}

\begin{figure}
\centering
\includegraphics[width=8cm]{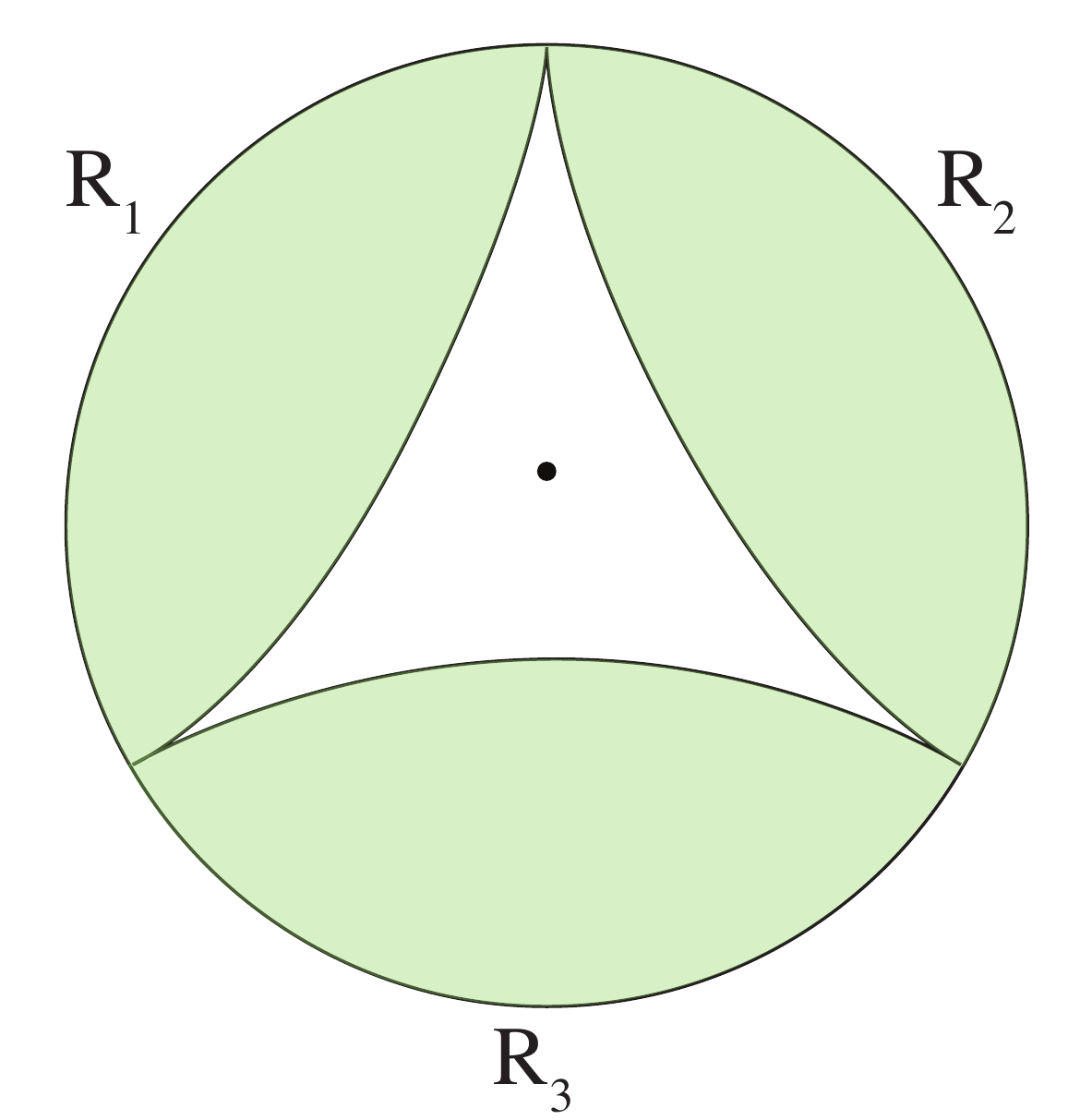}
\caption{The operator depicted in the center of this figure is not in $\CW(R_1)$,  $\CW(R_2)$, or $\CW(R_3)$. However, it does lie in the causal wedge of the union of any two regions $\CW(R_i \cup R_j)$ and can thus be written in terms of
boundary operators in the algebra of the combined regions.   
}
\label{fig_1}
\end{figure}

\begin{center}
\emph{Review of the DHW argument}
\end{center}

The fact that the entanglement wedge $\EW(R)$ is the ``largest'' bulk region that can be reconstructed from the algebra of $R$ will play a critical role in our work. For this reason, we will briefly review the arguments in
 \cite{Almheiri:2014lwa, DHW}, focusing especially on the aspects of this literature that will be the most relevant for the framework that we begin to develop in section \ref{sec_3}.

Suppose that we are given\footnote{While we take the code subspace as given, it should be possible to identify code subspaces purely from the CFT. For example, a necessary (but not sufficient) condition for  a collection of states to lie
in the same code subspace is that the collection has the property that subregions have entanglement entropies differing only by sub-leading contributions in $N$.}
 a particular code subspace $G \subset H$ which is known to be a span of states obtained by acting with a small number of low energy operators on a state where a semiclassical bulk exists; in particular, within $G$,
gravitational backreaction of bulk fields can be treated perturbatively. Dong, Harlow, and Wall (DHW) proved that if the support of an operator $\phi$ is contained in $\EW(R)$, then that 
operator can be reconstructed in $R$  \cite{DHW}. This means that there is an element of the algebra of $R$ whose action on states in the code subspace is the same as the action of $\phi$.

To understand the proof given in \cite{DHW}, we first refer to a result from quantum information. Refs. \cite{Almheiri:2014lwa, beny2007quantum, beny2007generalization}  show that if we have a code subspace G and some factorization of the full Hilbert space   $ G \subset H_{R} \otimes H_{\bar{R}}$, and if $\phi$ is some operator that acts within $G$ (it's action send states in the code subspace to other states in the code subspace), then
the following two statements are equivalent.
\begin{enumerate}
\item There exists an operator $O_R$ on  $H_R$ such that for any $\ket{\psi} \in G$,
\be
\phi \ket{\psi} = O_{R} \ket{\psi} \ \ \ \ \ \ \ \ \ \ \phi^\dagger \ket{\psi} = O^\dagger_{R} \ket{\psi}.
\label{recons1}
\ee
\item For any operator $X_{\bar{R}}$  on on  $H_{\bar{R}}$, we have 
\be
[\phi,X_{\bar{R}}]\big|_G = 0.
\label{com1}
\ee
\end{enumerate}

\noindent  While this theorem follows purely from quantum information, it plays a critical role in the entanglement wedge reconstruction argument. As suggested by the notation, we will associate $R$ with the factorization induced from boundary regions and G will be a code subspace with a semiclassical bulk interpretation. We can now discuss \cite{DHW}, which establishes that bulk semiclassical operators satisfy condition \ref{com1}, and the reconstructability follows because this is equivalent to \ref{recons1}.

We know the boundary Hilbert space can be factorized into a region and its complement $H =H_{R} \otimes H_{\bar{R}}$. For states with a semiclassical bulk interpretation, we can think about the extremal surface anchored to $\partial R$ as inducing its own tensor factorization of the code subspace G$_{\EW(R)} \otimes G_{\EW(\bar{R})}$.

Consider two states $\ket{\psi_0},\ket{\psi_1}\in G$ and the reduced density matrices obtained by tracing out the appropriate complement regions in the two factorizations 
\bn
\rho^0_{\bar R} &= \textrm{Tr}_R \ket{\psi_0} \bra{\psi_0} \nonumber \\
\rho^0_{\EW(\bar R)} &= \textrm{Tr}_{\EW(R) }\ket{\psi_0} \bra{\psi_0} 
\end{align}
Similarly, the density matrices $\rho^1_{\bar R}$ and $\rho^1_{\EW(\bar R)}$  are defined by the state $\ket{\psi_1}$.

The statement of a theorem in \cite{DHW} is that if the states satisfy:
\be
\rho^0_{\EW(\bar R)} = \rho^1_{\EW(\bar R)} \implies \rho^0_{\bar R}  = \rho^1_{\bar R} 
\label{density1}
\ee
then, an operator of the form $\phi = \mathbb{1} \otimes \phi_{\EW(R)}$ acting only within the entanglement wedge of $R$ will satisfy the two equivalent properties of \ref{recons1} and \ref{com1}.

To understand this, we note that the result in \cite{JLMS} established a precise relationship between the bulk and boundary modular hamiltonian. This provides the connection between the first equality and second equality in \ref{density1}. Now, the operator $\phi$ supported in the entanglement wedge of a boundary region $R$ does not affect the state in the complement wedge (this just follows from semiclassical field theory). Thus, if we define $\ket{\psi_1}$ as
\be
\ket{\psi_1} = e^{i\epsilon \phi}\ket{\psi_0}
\label{state1}
\ee
the first equality in \ref{density1} is satisfied. The second equality then implies that the expectation value of any operator in the algebra of $\bar{R}$ is the same in both states:
\be
\bra{\psi_0} X_{\bar{R}} \ket{\psi_0} - \bra{\psi_1} X_{\bar R} \ket{\psi_1} =0 
\ee
Rewriting the second term using \ref{state1} and expanding to first order in $\epsilon$ we obtain \ref{com1}.

This proves that within the code subspace, we can express operators in the entanglement wedge of $R$ in terms of operators in the algebra of $R$. Moreover, if an operator on $G$ has support outside $\EW(R)$, it must have
no reconstruction in $R$.  To see this, suppose that an operator $\phi$ on $G$ had support outside $\EW(R)$ so that it fails to commute with some operator $\phi^\prime$ on $\EW(\bar{R})$.  The argument above shows that there
exists a reconstruction $O^\prime_{\bar{R}}$ of $\phi^\prime$ that acts on $\bar{R}$. If $\phi$ could be reconstructed with an operator $O_R$ on $R$, we would have $[O^\prime_{\bar{R}}, O_R] =0$ which contradicts the fact that $[\phi^\prime,\phi]\neq0$. 
%

Our final conclusion is that an operator acting on a code subspace can be reconstructed in a region $R$ of the CFT if and only if its support in entirely contained in $\EW(R)$. 
By exploiting the reconstructability for states in the code subspace, we now explore how the bulk, including the conformal metric, is encoded in the CFT. 

We note that the reconstructability argument itself is a statement about a special class of quantum states and makes no reference to the plank length in the bulk. However, in making the connection between the reduced density matrix in the entanglement wedge \cite{JLMS} and the boundary, one clearly needs to assume some notion of locality. In particular, this involves taking $N \rightarrow \infty$.

\section{Superficially Local Operators}
\label{sec_3}

For the rest of this paper we work in the context of the ``infinite $N$ limit.''  It is assumed that there are code subspaces $\{G\}$ of the CFT Hilbert space $H$ that are
holographically dual to quantum field theory on (asymptotically AdS) spacetime backgrounds.    Setting $N=\infty$ in this way may cause discomfort, especially with some of the more
complicated things we discuss below, and for this reason we have
provided appendix \ref{app1} which defines our quantities while taking the large $N$ limit more carefully.   Even without reading the appendix, 
the majority of our development can made much more precise simply by replacing equalities with approximate equalities which, in the large $N$ limit, approach authentic equalities.  

In this section we are going to almost completely answer a fundamental question:  Suppose that a code subspace $G$ is given and that we are told that $G$ is dual
to some unknown field theory on some unknown spacetime background.  Let $\phi$ be a given operator on  $G$.  
Is  $\phi$ dual to a local operator? Note that we are given no information about $\phi$ (other than how it acts on $G$) and, in particular, it is probably not a local CFT operator.   
The ability to answer this question is equivalent to finding
all of the CFT operators that are dual to local bulk operators with respect to our particular code subspace.

Prior work addresses related issues but falls short of providing a general identification of local bulk operators. Consider, again, the HKLL method \cite{HKLL}.  If $\phi$ is a quantum field in the bulk $M$, then, given a point $p \in M$, it is possible to solve the field equation of motion and obtain an expression of the form
\begin{equation}
\label{HKLL}
\phi(p) = \left( \int_{\partial M} K(p,x) O(x) d^{D-1} x\right) \bigg{|}_{G}.
\end{equation}
Here, the boundary field $O$ is the one associated with $\phi$ through the extrapolate dictionary.  As discussed above, the integration kernel $K$ is not unique.  While different choices of $K$ yield different CFT operators,
 the restriction of these different choices of operators to the code subspace $G$ must always give the same answer.

At a first glance, equation \ref{HKLL} appears to not only identify the nonlocal CFT operators that are dual to local bulk operators, but even provides a formula for them.  This is not the case however.
The integration kernel can only be found by solving equations of motion on the curved spacetime background $M$, and this assumes knowledge of what the background is.  There are very few code subspaces
for which the corresponding geometry is known.  Another reason that the HKLL procedure is unsatisfactory for our purposes is that  it only identifies a subset of the boundary operators that
are dual to local bulk operators.  We would like to find a more general characterization of locality in the bulk at leading order in $1/N$.

\subsection{Comparing Locality of Operators}

\begin{figure}
\centering
\includegraphics[width=8cm]{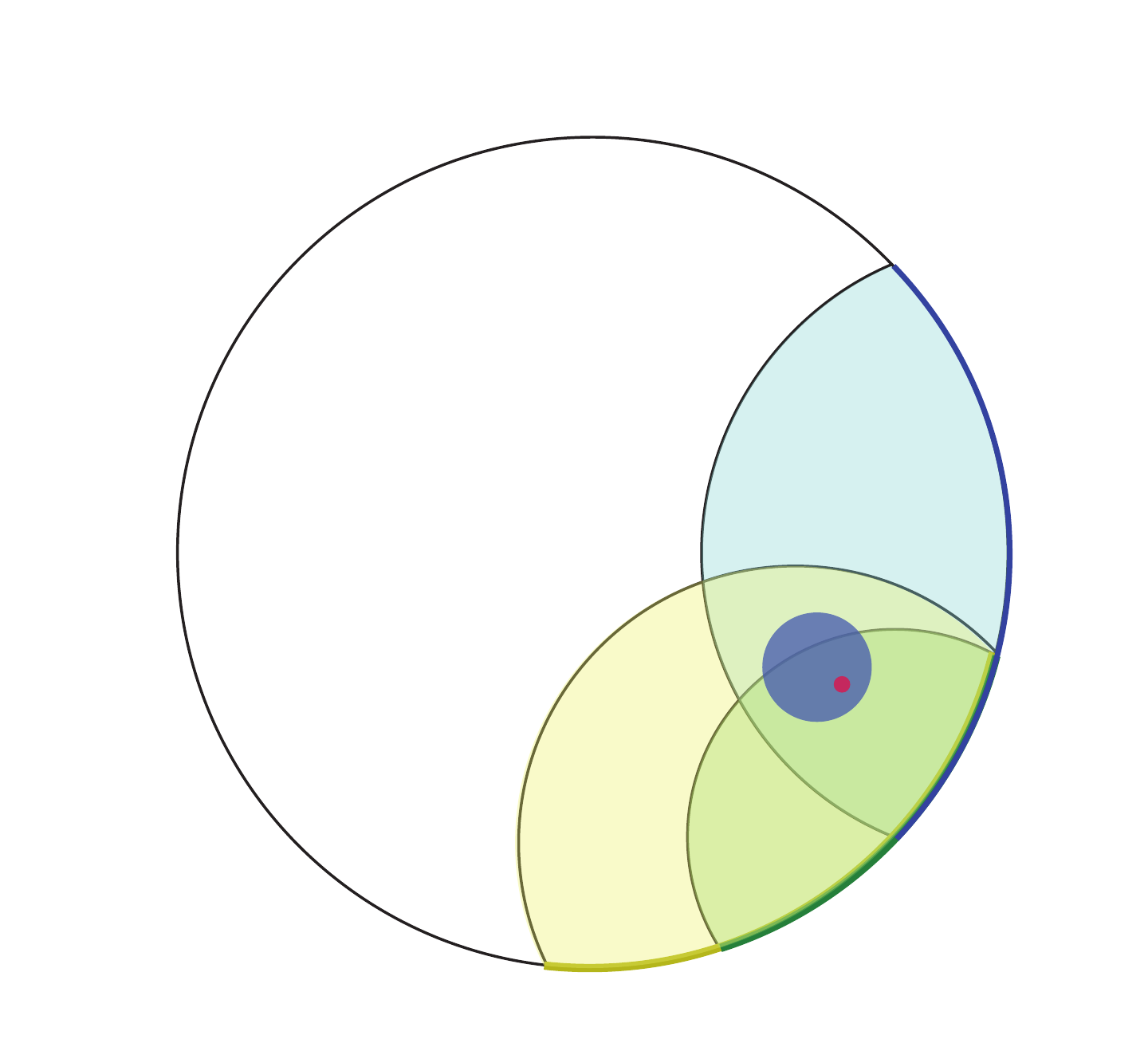}
\caption{
A nonlocal bulk operator $\phi_1$ will clearly lie in fewer regions than an operator $\phi_2$ whose support is entirely contained in the first $\mathcal{Q}(\phi_1) \subset \mathcal{Q}(\phi_2$) . 
}
\label{fig_2}
\end{figure}

Our guiding principle is that that, roughly speaking, the more local a bulk operator $\phi$ is, the more distinct boundary regions exist for which $\phi$ can be reconstructed.  This follows from
subregion duality as explained in section \ref{sec_2}.  To make this concept more precise, we are going to employ the full power of the quantum error-correcting structure of AdS/CFT to introduce 
a function $\mathcal{Q}$ that  maps operators on $G$ to the collection of all possible boundary regions that can reconstruct a given operator.  $\mathcal{Q}$ will then provide a measure of locality
of every operator.   We now explain this precisely.

Let $\mathcal{R}$ denote the collection of all $D-2$ dimensional achronal submanifolds of $\partial M$.  Informally, $\mathcal{R}$ is the collection of all regions $R$ upon which one would compute a
von Neumann entropy by anchoring stationary surfaces \cite{Ryu:2006bv,Hubeny:2007xt} to $\partial R$.  Note that we are not restricting to a single time slice of $\partial M$.  
If $R\in \mathcal{R}$ and $\phi$ is an operator that, along with its hermitian conjugate, acts on the code subspace $G$, then $\phi$ is said to be reconstructable in $R$
if there exists $O$ in the algebra of $R$ such that $O \big|_G = \phi$ and $O^\dagger \big|_G = \phi^\dagger$. We now give a critical definition:

\begin{definition}
Suppose that $\phi$ is an operator on $G$ and $R \in \mathcal{R}$.  Then,  we define
\[
\mathcal{Q}(\phi) = \{R \in \mathcal{R}  \: \big| \: \phi \textrm{ is reconstructable in } R  \}.
\]
\end{definition}

\noindent Whatever the (unknown) geometry of $M$ is, subregion duality (see section \ref{sec_2}) gives a geometrical condition for $\mathcal{Q}(\phi)$ to contain a region $R$.  Specifically, $R \in \mathcal{Q}(\phi)$
if and only if the (bulk) support\footnote{The support of an operator is defined as follows.
Let $A$ be a (possibly nonlocal) operator on a quantum field theory on the curved spacetime $M$.  Let $U$ be the set of points in $M$ such that for every point $p$ in $U$, every local bulk operator at $p$ commutes with $A$.
Then, the support of $A$, denoted by $\textrm{supp } A$, is given by  $M \setminus \left( J_+(U) \cup J_-(U) \right)$}  of $\phi$ is contained in the entanglement wedge of $R$.  This immediately implies the following properties of $\mathcal{Q}$:

\begin{prop}
\label{Q_properties}
Let $\phi_1$ and $\phi_2$ be two operators on the code subspace $G$.  Then,
\begin{enumerate}
\item If  $\mathrm{supp} \: \phi_1 \supseteq \mathrm{supp} \:  \phi_2$ , then $\mathcal{Q}(\phi_1) \subseteq \mathcal{Q}(\phi_2)$,
\item if $\mathrm{supp} \: \phi_1 = \mathrm{supp} \:  \phi_2$, then $\mathcal{Q}(\phi_1) = \mathcal{Q}(\phi_2)$.
\end{enumerate}
\end{prop}

Note that the converses to these statements, though seemingly desirable, are false in many cases.  This is somewhat disappointing: the bulk
support of an operator is a property of the operator's bulk description while $\mathcal{Q}$ is a function that is manifestly defined in the boundary theory.  Our goal is to find a ``boundary-only'' characterization of bulk locality, so we would be much better off if the converse to Proposition \ref{Q_properties}  were in fact true.

What $\mathcal{Q}$ does accomplish is that it identifies the support of an operator to the greatest possible ``resolution'' that the boundary theory can easily see.  For this reason we define an equivalence relation on operators on $G$: $\phi_1 \sim \phi_2$ if $\mathcal{Q}(\phi_1) = \mathcal{Q}(\phi_2)$.  We use the notation $[\phi]$ to denote the equivalence class of $\phi$ with respect to this relation.  In other words, $[\phi]=Q^{-1} (Q(\phi))$.  Two operators are in the same class if they
are ``the same as far as $\mathcal{Q}$ can tell.''  We can attempt to compare the locality of two operators by putting a partial ordering on the collection of equivalence classes by writing $[\phi_1] \leq [\phi_2]$ if $\mathcal{Q}(\phi_1)
\subseteq \mathcal{Q}(\phi_2)$ (which is a well-defined relation).  Note that a trivial operator like the identity on $G$, denoted by $1_G$, can be reconstructed in any region.  Thus, $[\phi] \leq [1_G]$ for any operator $\phi$ on $G$.

We are now ready to give a plausible characterization of a local bulk operator by means of $\mathcal{Q}$.

\begin{definition}
Suppose that $\phi$ is an operator on $G$.  $\phi$ is said to be \emph{superficially local} if
\begin{enumerate}
\item $[\phi] \neq [1_G]$ and
\item Every operator $\phi^\prime$ with the property that $[\phi] \leq [\phi^\prime]$ has $[\phi^\prime] \in \{ [\phi], [1_G] \}$.
\end{enumerate}
\end{definition}

\noindent We emphasize that the definition of a superficially local operator makes reference only to the boundary theory.  Thus, we can use this definition to offer an answer to the question posed above:
if we are given a large $N$ CFT with a Hilbert space $H$, a subspace $G$ of $H$, and an operator $\phi$, and if we told that $G$ is a code subspace corresponding to an unknown bulk spacetime,
then we can guess that $\phi$ is a local operator in the dual bulk theory if it acts on $G$ and if its restriction to $G$ is a superficially local operator.  This answer turns out to be right in many cases.

The word ``superficial'' is used for two reasons.  First, as we will shortly see,
there are examples of asymptotically AdS spacetimes for which some local bulk operators (for instance, those lying close to a spacelike singularity) are not superficially local.  Second,
we will not prove that every superficially local operator is local in the bulk.  The first of these deficiencies is completely unavoidable and it is tempting to contemplate its relation to the difficulties
of using AdS/CFT to describe points deep within a black hole interior \cite{Almheiri:2012rt} (although we will not pursue such contemplations here).  The second apparent deficiency is not a problem: in section \ref{sec_4} we will argue that it is possible  to identify when a given equivalence class of superficially local operators contains operators that are not actually local in the bulk.  This argument will be made in the boundary theory.  
The concept of superficial locality therefore provides a way to confidently identify a very large collection of operators on $G$ that should be interpreted as local operators in the bulk.  We now explain exactly which bulk operators can be found in this way.

\subsection{The Localizable Region}

As above, let $M$ be the asymptotically AdS bulk spacetime that is dual to a code subspace $G$ of a CFT in the large $N$ limit with Hilbert space $H$.  
In this section we are going to identify a special subset of $M$, denoted by $\Loc (M)$, which has the property that local bulk operators  at
points in $\Loc(M)$ can be successfully identified in the boundary theory through the consideration of superficially
local operators.

\begin{definition}
The \emph{localizable region} of $M$, denoted $\Loc(M)$, is the set of points $p \in M$ satisfying
\begin{enumerate}
\item If $\mathrm{supp \:}\phi = \{p\}$, then $\phi$ is superficially local and
\item if $\mathrm{supp \:}\phi = \{p\}$ and $[\phi^\prime] = [\phi]$, then $\mathrm{supp \:}\phi^\prime = \{p\}$.
\end{enumerate}
\end{definition}

\noindent Elements of $\Loc(M)$ will sometimes be called \emph{localizable points}.  Note that $\Loc(M)$ is a subset of the bulk and its definition makes reference to the concept of the bulk support of an operator, 
so this definition is not particularly transparent from the boundary theory.  However, a connection with the boundary theory
becomes apparent when $\Loc(M)=M$:
\begin{prop}
\label{loc_equals_suploc}
If $\Loc(M) = M$, an operator $\phi$ on $G$ is superficially local if and only if it is local in the bulk.  Moreover, if $\phi_1$ and $\phi_2$ are two superficially
local operators with $[\phi_1] = [\phi_2]$, then they must be local at the same bulk point.
\end{prop}
\begin{proof}
If $\phi$ is a local operator, the definition of $\Loc(M)$ immediately demands that $\phi$ is superficially local. Conversely, let suppose that $\phi$ is superficially
local.  If $\phi$ is not local in the bulk, then there are at least two distinct points $p$ and $q$ in the support of $\phi$.  Let $\phi^\prime$ be a local operator at $p$.
By Proposition \ref{Q_properties}, the fact that $\mathrm{supp \:} \phi^\prime \subseteq   \mathrm{supp \:} \phi$ means that $[\phi] \leq [\phi^\prime]$.  But $\phi$ is
superficially local and $\phi^\prime$ is nontrivial so we conclude that $[\phi] = [\phi^\prime]$.  The definition of the localizable region now demands that
$\mathrm{supp \:}\phi = \{p\}$, a contradiction.

Now suppose that $\phi_1$ and $\phi_2$ are two superficially local operators with $[\phi_1] = [\phi_2]$.  From what we just proved, we know that $\phi_1$ is local at some point, so the definition of the localizable region immediately demands that $\phi_1$ and $\phi_2$ are local at the same point.
\end{proof}

This result is a first step to providing a boundary description of $\Loc(M)$ because the notion of superficial locality is one of the boundary theory.
Unfortunately the hypothesis of Proposition \ref{loc_equals_suploc} is often too much to ask for.  To better understand this, consider the following result
which which establishes a geometrical bulk interpretation of $\Loc(M)$.\footnote{Theorem \ref{Loc_iff} elucidates the connection between our program and the ideas of \cite{Czech:2014ppa, Czech:2015qta, 
Balasubramanian:2013rqa,Balasubramanian:2013lsa,Czech:2014wka,  Headrick:2014eia}.  Note this work is primarily interested in the reconstruction of bulk geometry while our focus is on operator reconstruction.  However,
below in section \ref{sec_4} we will reconstruct aspects of the bulk geometry.}
\begin{thm}
\label{Loc_iff}
$p \in \Loc(M)$ if and only if there exists a subset $\mathcal{R}_0$ of the collection of boundary regions $\mathcal{R}$ such that
\[
\bigcap_{R\in \mathcal{R}_0} \EW(R) = \{p\}.
\]
\end{thm}
\begin{proof}
Suppose first that there exists $\mathcal{R}_0$ satisfying the condition given in the statement of the theorem.  Fix a local bulk operator $\phi$ 
at $p$ so that $\mathrm{supp \:} \phi = \{p\}$.  $\mathcal{Q}(\phi)$ must contain all regions $R$ with $p \in \EW(R)$ so, in particular, $\mathcal{R}_0 \subseteq \mathcal{Q}(\phi)$.
If $\phi^\prime$ is some operator on $G$ with $[\phi^\prime] \geq [\phi]$, then  $\mathcal{Q}(\phi) \subseteq \mathcal{Q}(\phi^\prime)$ so we have
\begin{align*}
\mathrm{supp \:} \phi^\prime   \subseteq \bigcap_{R\in \mathcal{Q}(\phi^\prime)}  \EW(R)  \subseteq \bigcap_{R\in \mathcal{Q}(\phi)} \EW(R)& \\
 \subseteq \bigcap_{R\in \mathcal{R}_0} \EW(R) = \{p\}.&
\end{align*}
This implies that $\phi$ is superficially local so the first condition for $p\in \Loc(M)$ is satisfied. If it happens that the operator $\phi^\prime$ above
satisfies $[\phi^\prime] = [\phi]$, our argument still applies and we must therefore have $\mathrm{supp \:} \phi^\prime \subseteq \{p\}$.  It is not possible
to have $\mathrm{supp \:} \phi^\prime = \emptyset$ since this would require that $[\phi^\prime] = [\phi] = \mathcal{R}$ which is false.  We conclude that
 $\mathrm{supp \:} \phi^\prime = \{p\}$ and thus that $p \in \Loc(M)$.
 
 We now prove the converse.  Let $p$ lie in $\Loc(M)$.  Suppose that there does not exist any $\mathcal{R}_0 \subseteq \mathcal{R}$ with
 $\bigcap_{R\in \mathcal{R}_0} \EW(R) = \{p\}$.  Let $\phi$ be a local operator at $p$ which requires that $\phi$ is superficially local.  There
 must exist a point $q \in M$ with
 \[
 q \in  \left(\bigcap_{R\in \mathcal{Q}(\phi)} \EW(R) \right) \setminus \{p\}.
 \]
 Now consider a local operator $\phi^\prime$ at the point $q$.  Since $q$ lies in the entanglement wedge of every region whose entanglement wedge contains $p$,
 we have $[\phi] \leq [\phi^\prime]$.  The superficial locality of $\phi$, along with the fact that $\phi^\prime$ is not trivial, implies now that $[\phi] = [\phi^\prime]$ which,
 by the definition of $\Loc(M)$, implies that $\mathrm{supp \:} \phi^\prime = \{p\}$ which is a contradiction.
\end{proof}

%

Theorem  \ref{Loc_iff} is a useful tool for identifying examples of localizable regions in asymptotically AdS spacetimes as we will do in section \ref{sec_examples}.
For now, we only advertise some facts that may be of interest.  Localizable regions can extend quite far into the bulk spacetime.  For the same reason that extremal
surfaces can penetrate event horizons in some cases, $\Loc(M)$ can intersect a black hole interior.  However, points that are too close to spacelike
singularities are not localizable.  Another interesting property of localizable regions is that they are not always subsets of the portion of the bulk that is accessible
to boundary-anchored extremal codimension 2 surfaces with minimal area.  In other words, $\Loc(M)$ can have a nonempty intersection with the entanglement shadow \cite{Balasubramanian:2014sra}.
Before discussing these examples, however, we are going to introduce an object that will greatly increase the motivation for studying the localizable region.

\subsection{The Space of Classes}
\label{sec_class_space}

The object that we now study is the collection of all equivalence classes of superficially local operators on $G$.  We suggestively denote this set by
$\tilde{M}$:
\[
\tilde{M} = \left\{ [\phi] \: \big| \: \phi \textrm{ is a superficially local operator on } G\right\}.
\]
Given that an element $P \in \tilde{M}$ is a set of operators, all of which have the same value of $\mathcal{Q}$, it is convenient to let define $\mathcal{Q}(P)$
as $\mathcal{Q}(\phi)$ for any choice of $\phi \in P$. 

An intuitive picture of $\tilde{M}$ is clear when $M = \Loc(M)$.  In this case, Proposition \ref{loc_equals_suploc} shows that there is a one-to-one correspondence
between $\tilde{M}$ and $M$.  The correspondence is that a point $p \in M$ is identified with the collection of all local operators at $p$.  This reveals a new approach
to bulk reconstruction from the boundary theory, somewhat similar in spirit to that of \cite{Engelhardt:2016wgb}, which we will explore below. 

Let us now make no assumptions about $\Loc(M)$ and determine the general structure of $\tilde{M}$.  What we are going to find is that $\tilde{M}$
is equal to $\Loc(M)$  with the possible addition
of some extra points in $\tilde{M}$.  We refer to these unwanted extra points as ``clumps.''

First suppose that $p \in \Loc(M)$ and let $\phi$ be a local bulk operator at $p$.  Then, $[\phi]$ consists only of local operators at $p$. (This follows
directly from the definition of the localizable region.)  As a consequence, a copy of $\Loc(M)$  can always be identified in $\tilde{M}$.  Another thing
that we can immediately show is that if $\Phi$ is any superficially local operator whose support  consists of more than one point, then $\mathrm{supp \:} \Phi
 \cap \Loc(M) = \emptyset$.  To see, this, suppose that $p \in \mathrm{supp \:} \Phi \cap \Loc(M) $ and consider a local operator $\phi$ at $p$.  We would
 then have $[\Phi] \leq [\phi]$ with $\Phi$ superficially local so  $[\Phi] = [\phi]$.  This contradicts the definition of $\Loc(M)$ since $\Phi$ is nonlocal.
 
 We cannot exclude the possibility that there exist nonlocal superficially local operators. To investigate this issue carefully, we introduce a map
 $C$ that sends a point $P$ in $\tilde{M}$ to a subset of $M$ as follows:
 \[
 C(P) = \bigcup_{\Phi \in P} \mathrm{supp \:} \Phi.
 \]
 $C$ has some nice properties:
 \begin{prop}
 \label{C_properties}
 Suppose that $P$ and $Q$ are elements of $\tilde{M}$.  Then,
 \begin{enumerate}
 \item If every element of $P$ is a local bulk operator, then there exists a point $p \in \Loc(M)$ such that $C(P) = \{p\}$,
 \item if $P$ contains a nonlocal operator, then $C(P) \cap \Loc(M) = \emptyset$,
 \item if $C(P) \cap C(Q) \neq \emptyset$, then $P=Q$ and, in particular, $C$ is injective.
 \end{enumerate}
 \end{prop}
 \begin{proof}
 1. If $P$ consists of only local operators, then all of those operators must be at the same bulk point. To see this,
 suppose that $\phi_1$ and $\phi_ 2$ are two local bulk operators at bulk points $p_1$ and $p_2$ respectively.
 Now $\mathcal{Q}(\phi_1) = \mathcal{Q}(\phi_2)$ so any linear combination $\alpha \phi_1 + \beta \phi_2$ must satisfy
 $[\phi_1] \leq [\alpha \phi_1 + \beta \phi_2]$.  The superficial locality of $\phi_1$ now proves that $[\phi_1] = [\alpha \phi_1 + \beta \phi_2]$
 which contradicts the assumption that $P$ consists only of local operators unless $p_1 = p_2$.  Now let $p$ denote the
 unique point in $M$ where the elements of $P$ are supported.  It is obvious now that $C(P) = \{p\}$.  Moreover, every local
 operator at $p$ must lie in $P$ and since there are no operators in $P$ with support beyond $\{p\}$ we conclude that $p \in \Loc(M)$. 

 2. Assume that $P$ contains a nonlocal bulk operator and suppose that $q \in C(P) \cap \Loc(M)$.  Let $\phi$ denote a local operator at $q$.
 There must be some operator $\Phi \in P$ with $q \in \mathrm{supp \:} \Phi$ so $[\Phi] \leq [\phi]$ from which the superficial locality of $\Phi$ 
 implies that $[\Phi] = [\phi]$ which is equivalent to the statement that $[\phi] \in P$.  But this means that $\phi$, a local operator in $\Loc(M)$,
 is equivalent to a nonlocal operator.  This is a contradiction.
 
 3.  Suppose that there exists a bulk point $x \in  C(P) \cap C(Q)$.  Let $\phi$ denote a local operator at $x$.  An argument identical to what was given
 for the proof of statement 2 shows that $\phi \in P$ and $\phi \in Q$.  But $P$ and $Q$ are equivalence classes so the fact that they share
 an element means that $P=Q$.
 \end{proof}

 This argument shows that $\tilde{M}$ can be thought of as the union of  $\Loc(M)$ with some extra points.  Each extra point $P$ has the property
 that $C(P)$ is a subset of $M$ with more than one element.  These objects are subtle enough to deserve a name:
 
 \begin{definition}
 Suppose that $P \in \tilde{M}$ has the property that $C(P)$ has more than one element.  Then, we will
 call both $P$ and $C(P)$ a \emph{clump}. 
 \end{definition}
 
 \noindent Clumps are somewhat problematic because both local and nonlocal operators in clumps are 
 superficially local.  They therefore represent a potential threat to our approach.  However, there is good news:
  we will argue in section \ref{sec_clump_conj} that clumps can be identified
 and removed using only the boundary theory (e.g. without relying on concepts like the bulk support of operators).
 Roughly speaking, clumps are associated with phase transitions for holographic entanglement entropy,
 and such phase transitions are visible in the boundary. 
 
 We are now in a position to give a much stronger answer to the fundamental question posed at the beginning of this section about identifying the operators
 on $G$ that are dual to local operators in the bulk.

 \medskip
\begin{thm} 
If there are no clumps, an operator $\phi$ on the code subspace $G$ is dual to a local bulk operator in the localizable region if and only if $\phi$ is superficially local.
\end{thm}
  
 \noindent If we assume the clump conjecture of section \ref{sec_clump_conj}, which provides a way to identify and eliminate clumps, this conclusion provides the boundary dual to the concept
 of a bulk local operator (within a certain region of the bulk).

\subsection{Examples}
 \label{sec_examples}
 
Examples can greatly clarify the machinery we have been developing.  In particular, the spacetimes below demonstrate several features:
 \begin{itemize}
 \item Despite being associated with HRT surfaces, $\Loc(M)$ can probe entanglement shadows.
 \item $\Loc(M)$ can intersect black hole interiors (but it does not extend arbitrarily closely to spacelike singularities).
 \item In regions that are close to spacelike singularities, local operators are not superficially local.
 \item Clumps can occur, but the only known examples are associated with phase transitions where extremal surfaces ``jump'' around them.
 \end{itemize}
 
 \begin{center}
 
  \emph{Vacuum AdS}
 
 \end{center}
 
 \noindent The simplest example is when $M$ is vacuum AdS space (or any small perturbation of vacuum AdS) with dimension $D \geq 2+1$.
 For any point $p\in M$, theorem \ref{Loc_iff} immediately shows that $p \in \Loc(M)$.  This is because in AdS space, one can always construct
 $D-1$ codimension $2$ stationary surfaces intersecting $p$, whose tangent spaces at $p$ are pairwise distinct,
 and then find the corresponding boundary regions $R_1, \ldots R_{D-1}$ on which these stationary surfaces are anchored.  To prove that $p \in \Loc(M)$, we then consider the collection
 of regions $\{R_1, \ldots R_{D-1}, \bar{R}_1 , \ldots , \bar{ R}_{D-1} \}$  and apply this set to theorem \ref{Loc_iff}.

\emph{Conclusion}: If we somehow know that $G$ is dual to a spacetime close to vacuum AdS, then an operator on $G$ is local if and only if it is superficially local.  The space of classes of superficially local operators, $\tilde{M}$, 
is a reconstruction of the bulk.
 
\begin{center}
 \emph{Conical AdS}

 \end{center}

\begin{figure}
\centering
\includegraphics[width=8cm]{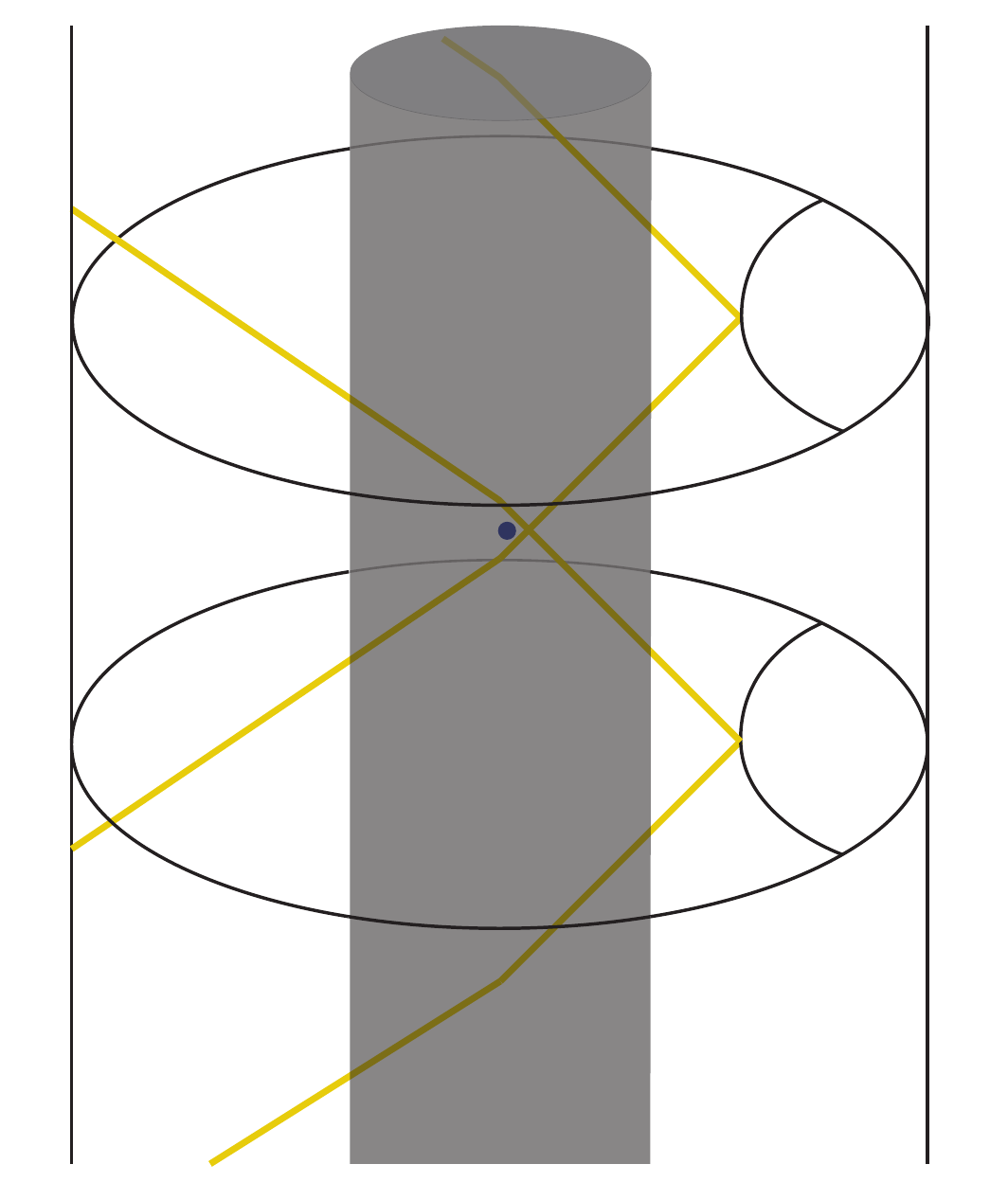}
\caption{Conical AdS is an example of how points in the bulk that are not directly probed by extremal surfaces can still be in the localizable region. Despite the entanglement shadow (the grey cylinder),
points can be localized because they can intersect boundaries of entanglement wedges. }
\label{fig_conical}
\end{figure}

\noindent Anti-de Sitter space with a conical deficit is a simple example of a spacetime with an entanglement shadow\footnote{To our knowledge, \cite{Balasubramanian:2014sra} and related work has
only studied regions that are not probed by minimal surfaces  anchored to static boundary regions rather than the general stationary surfaces appearing in the calculation of covariant holographic entanglement entropy.  Below 
we assume that the general features of the entanglement shadow in standard examples are unchanged if non-static surfaces are considered.} \cite{Balasubramanian:2014sra}.  Given that $\Loc(M)$ can be defined by means of
HRT surfaces, one might suspect that for conical AdS, $\Loc(M)$ is a proper subset of $M$. We will explain why  this is not the case and that, in fact, we again have $\Loc(M) = M$.

Let $n$ be an integer greater than $1$ and consider, for example, $M = \mathrm{AdS}_{2+1} / \bold{Z}_n$.  The metric can be written as
\begin{equation}
ds^2 = - \left({1 \over n^2} + {r^2 \over L^2} \right) dt^2 +  \left({1 \over n^2} + {r^2 \over L^2} \right)^{-1} dr^2 + r^2 d \phi^2
\end{equation}
where $-\infty < t < \infty$, $r>0$, and $\phi \in [0, 2\pi)$.  There is a critical radius $r_\mathrm{crit}$ such that
no HRT surface intersects the region $r<r_\mathrm{crit}$.  If $\{R_s\}$ is a continuous nested family of
boundary regions with $R_{-1}$ a small region and $R_1$ wrapping around almost the entire boundary, the HRT surface anchored to $R_s$, $\ext R_s$ , will discontinuously
jump around the shadow at some critical value of $s$.    Note that this phenomenon is not related to extremal surface barriers \cite{Engelhardt:2013tra} but is instead
a consequence of there being more than one stationary codimension 2 surface anchored to any given boundary region: no HRT surface enters the shadow because there would always be another
stationary surface that does not enter the shadow with smaller area.  The discontinuous jump can be regarded as a phase transition in the sense that the von Neumann entropy $S(R_s)$, regarded
as a function of the parameter $s$,  has a discontinuous derivative at the jump.

If $p \in M$ lies outside of the entanglement shadow, we must have $p \in \Loc(M)$ for the same reason that every point is localized in vacuum AdS.  On the other hand, suppose that $p$ lies within the entanglement
shadow.  To show that  $p \in \Loc(M)$, all we need, by theorem \ref{Loc_iff}, is a finite set of boundary regions such that the intersection of their entanglement wedges is $\{p\}$.

This can by done by considering regions like those shown in figure \ref{fig_conical}.  Note that only two regions are shown in the figure but that the point can be completely localized by
adding other boundary regions such as rotations of the regions depicted.  The trick here is easy to understand: it is not necessary for HRT surfaces to intersect localized
points as long as boundaries of entanglement wedges intersect them instead.

\emph{Conclusion}: If $G$ is dual to a spacetime close to $\mathrm{AdS}_{2+1} / \bold{Z}_n$, then an operator on $G$ is local if and only if it is superficially local. The space of classes of superficially local operators, $\tilde{M}$, is a 
reconstruction of the bulk.

\begin{figure}
\centering    
\subfigure{\label{fig:a}\includegraphics[width=70mm]{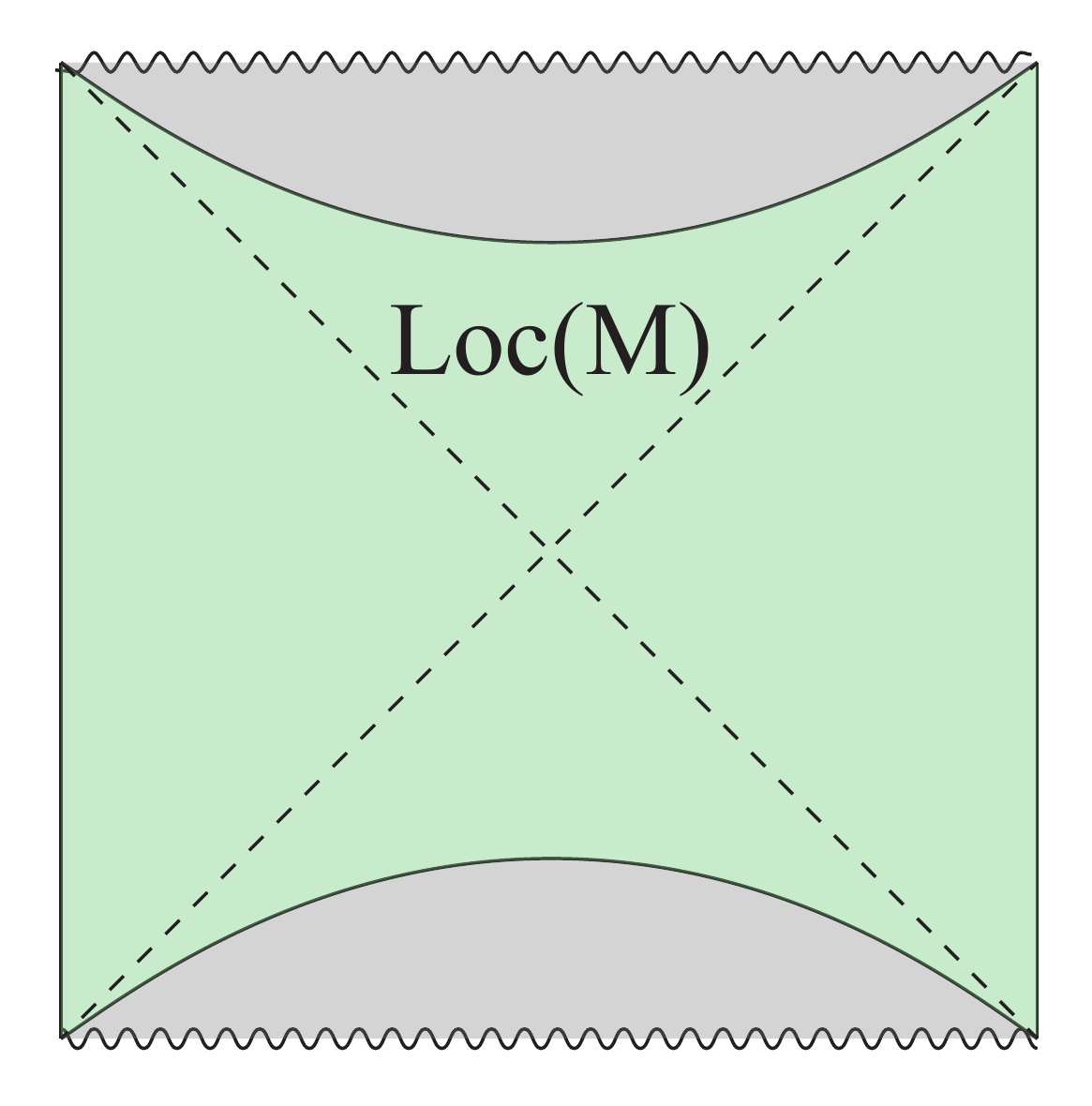}}
\subfigure{\label{fig:b}\includegraphics[width=70mm]{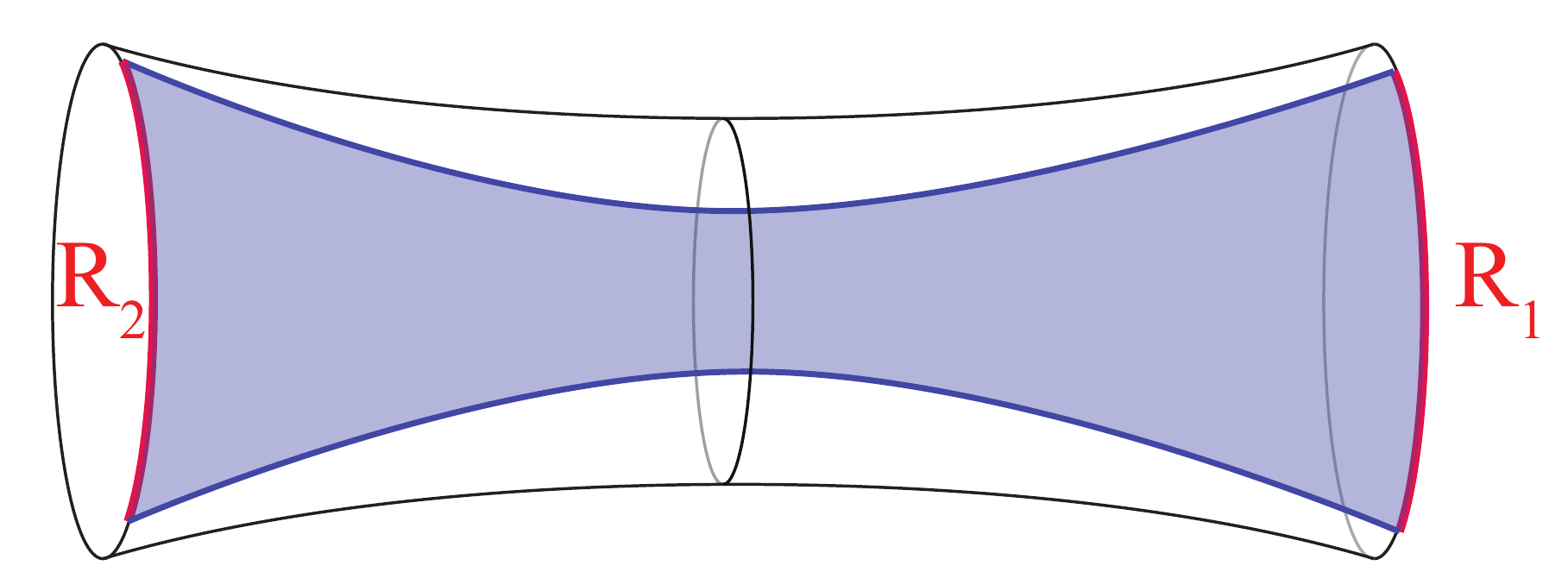}}
\caption{}
\label{fig_eternal}
\end{figure}

\begin{center}
 
\emph{Two-Sided Black Holes}

\end{center}

\noindent In the case where $M$ is an eternal AdS-Schwarzschild geometry, which has two disconnected boundary components, the localizable region extends into the black hole interior but  does not probe all the way to the singularity.  
This is depicted in figure \ref{fig_eternal}.  Many points in the interior region can be localized by considering boundary regions that consist of two disconnected components lying in different boundaries (see figure \ref{fig_eternal}).  HRT surfaces, however, do not  reach points that are arbitrarily close the future or past singularities: there is a critical radius $r_\mathrm{crit}$ (smaller than the black hole radius) that no 
boundary-anchored extremal surface extends beyond \cite{Wall:2012uf,Engelhardt:2013tra}.  Figure \ref{fig_singularity} proves that local operators at points with radius $r< r_\mathrm{crit}$ are not superficially local.  
This portion of the spacetime is completely missed by our methods and will thus be called the \emph{inaccessible region}.

\emph{Conclusion}: If $G$ is dual to an eternal AdS-Schwarzschild geometry (with two boundary CFTs), then an operator $\phi$ on $G$ is superficially local if and only if it is dual to a local bulk operator at a bulk point with $r>r_\mathrm{crit}$.
The space of classes of superficially local operators, $\tilde{M}$, is a reconstruction of the region of $M$ with $r>r_\mathrm{crit}$.

\begin{center}
 \emph{Dynamical Black Holes}

 \end{center}

 \noindent The previous example might have given the impression that $\Loc(M)$ cannot intersect a black hole interior without appealing to entanglement between two CFTs.  This is not the case. Consider a black hole that forms from
 collapse in an asymptotically AdS spacetime.  Then, it has been demonstrated \cite{Hubeny:2013dea} that HRT surfaces probe the black hole interior (although they do not
 approach the singularity arbitrarily closely).  Because such HRT surfaces can be anchored to boundary regions at a variety of angular positions,
 we conclude that $\Loc(M)$ enters the black hole interior in this case.  Note, however, that figure \ref{fig_singularity} again explains why regions too close to the singularity  are not localizable.
 
 \begin{figure}
\centering
\includegraphics[width=8cm]{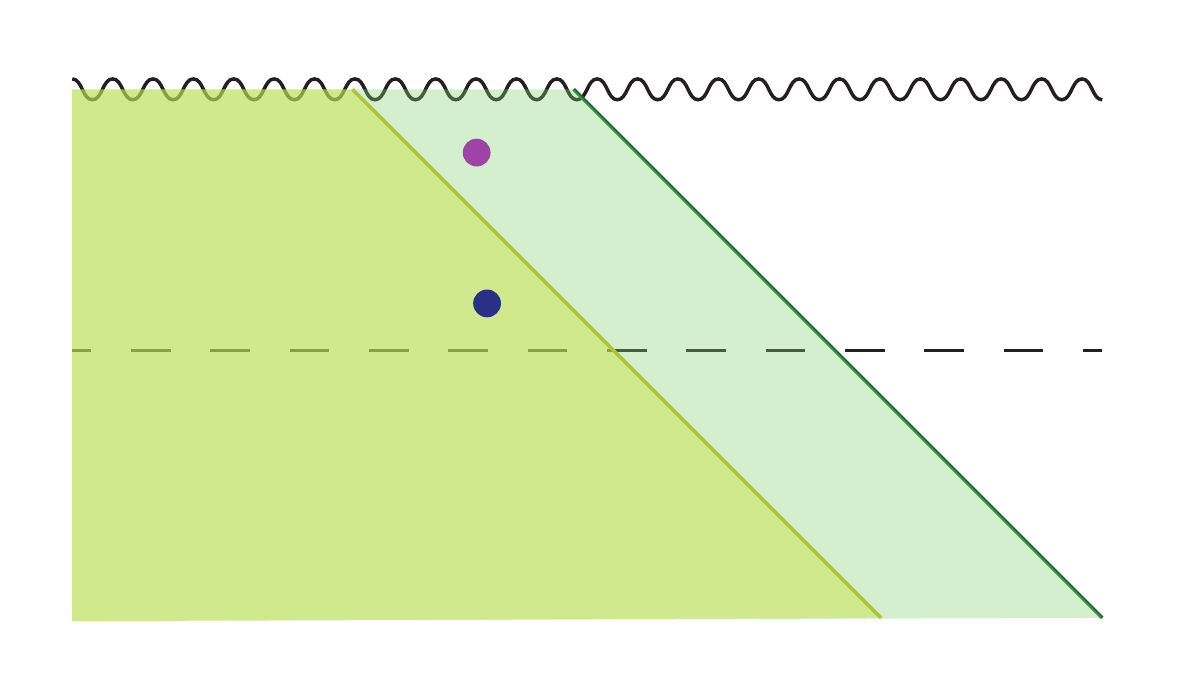}
\caption{
When a point (purple) is close to a spacelike singularity, it is very difficult for the point to be in $\Loc(M)$.  Quite generally, HRT surfaces are prevented from approaching such singularities
 \cite{Wall:2012uf,Engelhardt:2013tra}.  In this figure, the horizontal dashed line is a surface with the property that no HRT surface intersects its future. (This is more restrictive than an
 extremal surface barrier, which would prohibit smooth deformations of stationary surfaces.)
  A local operator at the purple point cannot be superficially local since
a point in its past (blue) will typically be contained in strictly more entanglement wedges.
}
\label{fig_singularity}
\end{figure}

\begin{figure}
\centering    
\subfigure{\label{fig:a}\includegraphics[width=70mm]{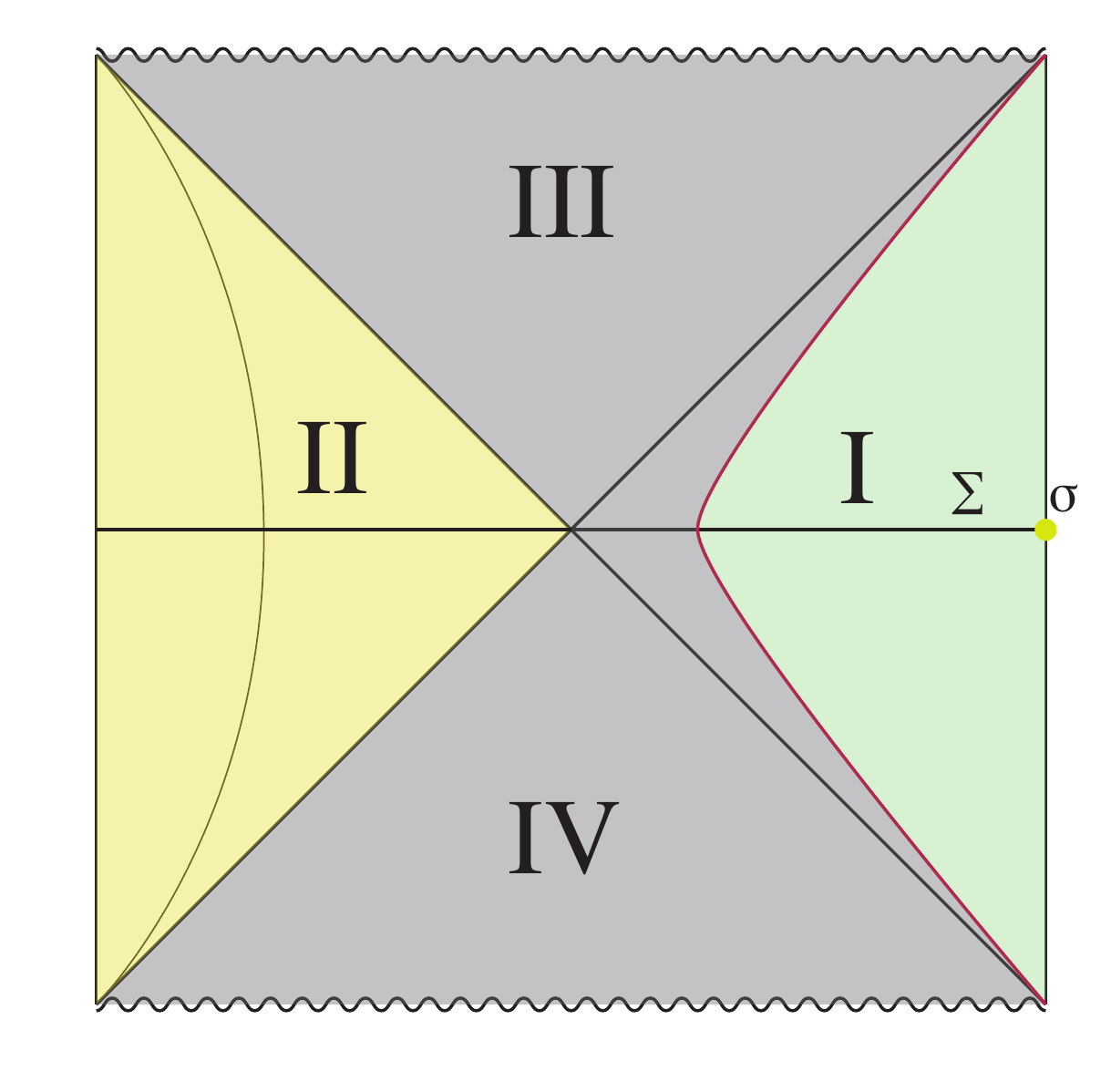}}
\subfigure{\label{fig:b}\includegraphics[width=70mm]{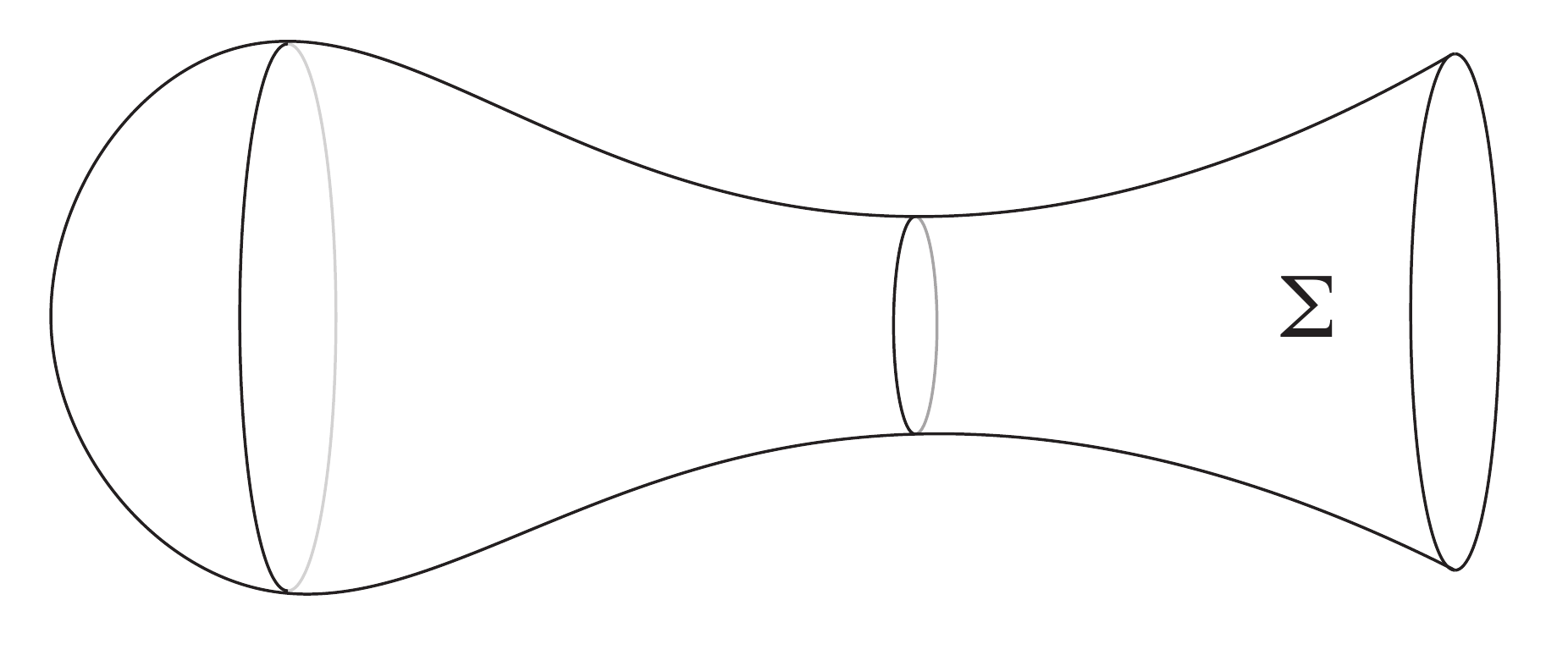}}
\caption{The bag of gold geometry we consider is obtained by removing an asymptotic region from an AdS black hole and replacing it with a patch of de Sitter space.  As discussed in the text, the localizable region  is the portion of region I
that is accessible to HRT surfaces and region II is a single clump.
The remaining portion of the spacetime is ``inaccessible'' in the sense that no operator with support in these regions is superficially local.}
\label{fig_bog}
\end{figure}
%

\begin{center}
 
\emph{Bag of Gold}

\end{center}

\noindent Our fourth example is a ``bag of gold'' spacetime (see, e.g., \cite{Marolf:2008tx}).  The manifold $M$ is an AdS-Schwarzschild spacetime with one of its two asymptotic regions removed and replaced with a patch of de Sitter space.
 The spacetime is static and spherically symmetric.  Its Penrose diagram is shown in figure \ref{fig_bog}.  We will label the regions in the diagram I-IV as shown in the figure  (note that region II includes
 the de Sitter patch).
 It is very important to understand that unlike the two-side AdS-Schwarzschild spacetime, $M$ has only one asymptotic boundary with topology $S^{D-2} \times \bold{R}$.
 The time slice $\Sigma$ that is marked in figure \ref{fig_bog}  has the topology of $\bold{R^{D-1}}$.  In particular, $\Sigma$ is simply connected and the homology
 constraint for HRT surfaces will not play any interesting role here.  
 The dotted line in region I is a surface beyond which no HRT surface probes.

We will argue the following.  
\begin{enumerate}
\item $\Loc(M)$ is the portion of region I that is probed by HRT surfaces.
\item $\tilde{M}$ has a single clump  whose image under $C$ (see section \ref{sec_class_space}) is all of region II. Thus, we will say that region II is a clump.
\item The rest of the spacetime (including regions II and IV) is neither localizable nor within clumps. It is ``inaccessible.''
\end{enumerate}

First let us discuss why region II is a clump.  Like conical AdS, this spacetime exhibits phase transitions in its HRT surfaces as well as an entanglement shadow.
Consider the boundary time slice $\sigma = \partial \Sigma$ and let $R_\psi$ be a spherical cap on $\sigma$ with opening angle $\psi$ (defined so that $R_\pi = \sigma$).
The spacetime in region I is identical to region I of AdS-Schwarzschild  so the structure of stationary codimension 2 boundary-anchored surfaces must also be the same and, in particular, there are always two distinct stationary surfaces anchored to $R_\psi$.  At $\psi = \pi/2$, there is a phase transition with a discontinuity in
the first derivative of $S(R_\psi)$.  At this transition, the minimal surface jumps around the entire region II.  Note also that HRT surfaces fail to even contact the bifurcation throat: there is, once again,
a minimal radius in region I, $r_\mathrm{crit}$, greater than the black hole radius, within which no HRT surface extends.

If $\psi < \pi/2$, the spatial region $V_\psi$ on $\Sigma$ between $R_\psi$ and its HRT surface $\ext R_\psi$  is confined to region I. Thus $\EW(R_\psi)$ is confined to
region I; this follows from the fact that $\EW(R_\psi)= D(V_\psi)$ after compactification.  Meanwhile, When $\psi> \pi/2$, $V_\psi$ contains the entire intersection of $\Sigma$ with region II
and $\EW(R_\psi)$  must contain all of region II.
These observations were made for a simple spherical cap on the time-reversal symmetric slice $\Sigma$, but they hold very generally: any time we consider a nested family of boundary regions
$\{ R_s \in \mathcal{R} \}$, $\EW(R_s)$ is confined to region I for $s$ smaller than some critical value and $\EW(R_s)$ contains all of region II when $s$ exceeds this value.

What this shows is that if $\phi_1$ and $\phi_2$ are two bulk operators with support in region II, we must have $\mathcal{Q}(\phi_1) = \mathcal{Q}(\phi_2)$.  
Moreover, note that any operator $\phi$ which is supported in region II must be superficially local.  To see this, consider any $x \in M \setminus (\mathrm{\: region \: II})$.
If $x$ is in region I or III, take a spherical cap like $R_{\psi}$ with $\psi>\pi/2$, but place it on a boundary time slice at very early time.  No matter how early time time is taken,
time-translation invariance guarantees that region II is still contained in $\EW(R_\psi)$, but by sending the boundary time slice to $-\infty$, we can put any point in regions I or III in the future
of $\ext R_\psi$.  This means that there exists some $R \in \mathcal{Q}(\phi) \setminus \mathcal{Q}(\phi^\prime)$ so $\phi \nleq \phi^\prime$.  The same argument can be made if $x$ is in region IV
by sending the boundary time slice to $+\infty$.   We conclude that $\phi$ must be superficially local and region II is thus a clump (since all operators with support in region II are superficially local
and have the same image under $\mathcal{Q}$).

%
%

Let us finally study the remainder of the spacetime.  The portion of region I that is probed by HRT surfaces is readily seen to be contained in $\Loc(M)$.  We now outline an argument that, in fact, this probed region is exactly $\Loc(M)$.
Figure \ref{fig_singularity} gives an explanation of why local operators in region III cannot be localized.  More generally,
consider a local bulk operator $\phi_x$ at a point $x$ lying outside of the region probed by HRT surfaces but also lying outside of the clumped region II.  If $R$ is a boundary region with $x \in \EW(R)$, then $R$ must be large
enough to have undergone a phase transition so that region II is contained in the entanglement wedge of $R$ as well.  This means that if $\Phi$ is any superficially local operator in the clump, we have $\mathcal{Q}(\phi_x) \subsetneq
\mathcal{Q}(\Phi)$.    This shows that $\phi_x$ cannot be superficially local.

%
%

 \emph{Conclusion}: Suppose that $G$ is dual to the bag of gold geometry.  If $\phi$ is a superficially local operator, then it is either a local operator in the portion of region I probed by extremal surfaces or it is
 some operator (which need not be local) with support
 in region II.  The clump conjecture of section \ref{sec_clump_conj} is valid for this spacetime, so the problematic superficially local operators can be identified and discarded.
  After doing so, the remaining superficially local operators exactly form the collection of all bulk local operators in $\Loc(M)$.

%
%

\subsection{The Clump Conjecture\footnote{Section \ref{sec_clump_conj} presents material that significantly strengthens our conclusions.  However,
it may be distracting to focus on such technicalities and some may choose to pass over this section.}}

\label{sec_clump_conj}
In this section we propose a way to use the boundary theory to identify and remove clumps from $\tilde{M}$.  Specifically we give an alternative
definition of a clump that does not make direct reference to the bulk and we conjecture that our two definitions are equivalent.  We know of no counterexamples
to the conjecture and there is good evidence for its general validity.

The basic motivation is as follows.  If $P \in \tilde{M}$ is a clump, then, by definition, $C(P)$ contains more than one bulk point.  Generically, clumps have nonzero spacetime volume.
On the other hand, we know that no entanglement wedge can contain only part of a clump: if $R \in \mathcal{R}$, then either $C(P) \subseteq \EW(R)$ or $C(P) \cap (\EW(R))^\circ = \emptyset$. 
These observations indicate that if $R_s$ is a continuous nested one-parameter family of regions in $\mathcal{R}$ such that $R_s \in \mathcal{Q}(P)$ for $s>0$ and $R_s \notin \mathcal{Q}(P)$ when $s<0$,
we must have some form of a discontinuity in the entanglement wedges $\EW(R_s)$ as a function of $s$ at $s = 0$.  Such discontinuities occur when the HRT surfaces anchored to $\{R_s\}$ jump discontinuously.
But such a jump can often be seen in the boundary theory in the form of a discontinuity in a derivative of the von Neumann entropy of the boundary regions $R_s$.

Before stating the conjecture formally, we give a useful definition:

 \begin{definition}
Let $\phi$ be an operator on $G$ and $R \in \mathcal{Q}(\phi)$.
$R$ is said to be \emph{minimal} if whenever $R^\prime \subsetneq R$, $R \notin \mathcal{Q}(\phi)$.
 \end{definition}
 
 \noindent We will also introduce the map $\bar{\mathcal{Q}}$ by letting $\bar{\mathcal{Q}}(\phi)$ denote the collection of minimal elements of $\mathcal{Q}(\phi)$.
 Additionally, if $P \in \tilde{M}$, we will define $\bar{\mathcal{Q}}(P)$ as $\bar{\mathcal{Q}}(\phi)$ for any choice of $\phi \in P$ (all choices of $\phi$ have the same $\bar{\mathcal{Q}}(\phi)$).
 
 As suggested above, phase transitions in the boundary theory will play a role in the boundary identification of clumps.  To be clear, a ``phase transition'' refers to the following situation.
 Suppose that $\{R_s \big| -1<s<1 \}$ is a regular\footnote{By ``regular'' we mean that $R_s$ deforms smoothly enough that we are not introducing discontinuities in any derivative of von Neumann entropy
 by choosing an awkward parameterization of regions.} one-parameter
 family of boundary regions with $R_{s_1} \subsetneq R_{s_2}$ whenever $s_1 < s_2$.  Let $S(R_s)$ denote the von Neumann entropy of the boundary region $R_s$
 in any state\footnote{$S(R_s)$ is state-dependent, but the spacetime background is approximately fixed within the code subspace
 $G$, so assertions about phase transitions will be state-independent at leading order.} in the code subspace $G$.  We say that there is a phase transition at $s=0$
 if some derivative of $S(R_s)$ at $s=0$ is discontinuous.  Moreover, if $R \in \mathcal{R}$, we will say that there is a phase transition at $R$ if there is some one parameter
 deformation of the form above, $\{R_s\}$, with $R_0 = R$.

We now state our proposal for identifying and removing clumps.  We will refer to it as the \emph{clump conjecture}:
 
 \vspace{.4 cm}
\noindent \emph{Suppose that $P \in \tilde{M}$.  $P$ is a clump if and only if for every $R \in \bar{\mathcal{Q}}(P)$, there is a phase transition at $R$.}
 \vspace{.3 cm}
 
\noindent We immediately note that this conjecture is consistent with the examples provided in section \ref{sec_examples}.  
The only example we gave of a clump is that of the bag of gold spacetime which always features phase transitions for minimal regions.  Consider, however
the example of $\mathrm{AdS}_{2+1} / \bold{Z}_n$.  This may appear to contradict the clump conjecture because it is a spacetime with no clumps but which does posses 
phase transitions.  However, consider regions like the ones depicted in figure \ref{fig_conical}.  These are indeed minimal regions for the operator at the point depicted (which
corresponds to a point in $\tilde{M}$.  However,
there is no phase transition at such a region.  This is why the statement of the clump conjecture requires that there is a phase transition for \emph{every} $R \in \bar{\mathcal{Q}}(P)$.

\section{Reconstruction of Causal Structure and Beyond}
\label{sec_4}

From here on we assume the validity of the clump conjecture (which we strongly expect) and use a new definition of $\tilde{M}$:
\[
\tilde{M} = \left\{ [\phi]  \: \big| \: \phi \textrm{ is superficially local and } [\phi] \textrm{ is not a clump} \right\}
\]

\noindent This can be done using only the boundary theory. Simply
begin with $\tilde{M}$ as defined previously, and then remove clumps from it by using the clump conjecture.

With this new definition, a major conclusion of section \ref{sec_3} is that in some sense $\tilde{M}$ is isomorphic to $\Loc(M)$
although we have not been very clear about what sort of isomorphism this is.  We are now going
to take the view that $\tilde{M}$ can be thought of as a reconstruction of the bulk very seriously.  
We will successfully determine a metric on $\tilde{M}$ up to a conformal rescaling.  This will be done using only information
available in the boundary theory (which includes the definition of $\tilde{M}$ itself).  The manifold $\tilde{M}$ and its causal
structure will exactly reproduce that of $\Loc(M)$. This constitutes a boundary reconstruction of the metric on $\Loc(M)$ up to its conformal factor.

\subsection{Spacelike Separation and Microcausality}

The key insight to identifying a causal structure on $\tilde{M}$ is to note that $\tilde{M}$ consists of collections of operators on the code subspace $G$
and that the commutation relations amongst those operators must betray an aspect of the bulk spacetime geometry.  This suggests the following definition:

\begin{definition}
Suppose that $P,Q \in \tilde{M}$.  We say that $P$ and $Q$ are \emph{spacelike separated} if for every $\phi_1 \in P$ and  $\phi_2 \in Q$,  we have $[\phi_1,\phi_2] = 0$.
Otherwise, we say that $P$ and $Q$ are \emph{causally related}.
\end{definition}

There are two things to immediately notice about this definition.  First, while we have defined the statement that $P$ and $Q$ are causally related, we have not
yet  given meaning to the statement that $P$ is to the future of $Q$.  This will be addressed below.  Second, note that for $P$ and $Q$ to be causally related, all that is necessary is that
there exists some $\phi_1 \in P$ and some $\phi_2 \in Q$ such that $\phi_1$ and $\phi_2$ fail to commute.  It is certainly not necessary that all such operators
would fail to commute.

In special cases, it is possible to conclude that $P$ and $Q$ are spacelike separated without relying directly studying the commutativity of their operators.  If it happens that there exists
$R_1 \in \mathcal{Q} (P), R_2 \in \mathcal{Q}(Q)$ with the property that $R_1$ and $R_2$ are spacelike separated in the boundary, meaning that 
\[
\left( J_+^\partial(R_1) \cup  J_-^\partial(R_1)\right) \cap R_2 = \emptyset,
\]
then microcausality in the boundary field theory guarantees that any operators $O_1$ and $O_2$ in the algebras of $R_1$ and $R_2$ respectively must have $[O_1,O_2]=0$.  In particular,
for any  $\phi_1 \in P$ and  $\phi_2 \in Q$, we can find reconstructions of $\phi_1$ and $\phi_2$ in $R_1$ and $R_2$ respectively and conclude that $[\phi_1,\phi_2] = 0$.  However,
this situation is too much to ask for in general.  

In the case where two classes $P$ and $Q$ are causally related, the above logic indicates that there absolutely cannot be any 
$R_1 \in \mathcal{Q} (P), R_2 \in \mathcal{Q}(Q)$ with the property that $R_1$ and $R_2$ are spacelike separated in the boundary.
This is consistent with a theorem in bulk geometry which is a necessary result for the consistency of entanglement wedge reconstruction:
\begin{prop}
Let $M$ be an asymptotically AdS spacetime and suppose that $p,q \in M$ are bulk points with $q \in I_+(p)$.  Suppose, moreover, that there exist
boundary regions $R_1, R_2 \in \mathcal{R}$ such that $p \in \EW(R_1), q \in \EW(R_2)$. Then, $\left( I_+^\partial(R_1) \cup  I_-^\partial(R_1)\right) \cap R_2 \neq \emptyset$.
\end{prop}
\begin{proof}

Choose a Cauchy surface $\sigma$ of $\partial M$ with $R_1 \subseteq \sigma$ and let $\bar{R}_1 = \sigma \setminus R_1$. 
Let $\Sigma$ be any AdS-Cauchy surface for the bulk with $\partial \Sigma = \sigma$ and write $\Sigma = S \cup \bar{S}$ where 
$S \cap \bar{S}$ is the HRT surface of $R_1$.  Then, $q \notin \EW(\bar{R}_1)$.  (This follows from
the fact that $\EW(R_1) = D(S)$ and $\EW(\bar{R}_1) = D(\bar{S})$.)

Suppose that we had $R_2 \subseteq D^\partial(\bar{R}_1)$.  Wall's entanglement wedge nesting theorem \cite{Wall:2012uf} implies that this would
require that $\EW(R_2) \subseteq \EW(\bar{R}_1)$ which contradicts the fact that $q \in \EW(R_2)$.  Thus, $R_2$ is not contained (entirely) in
$D^\partial(\bar{R}_1)$.  On the other hand, the boundary is flat  so $D (\bar{R}_1 )= \partial M  \setminus ( I_+^\partial (R_1) \cup I_-^\partial(R_1))$.
We conclude that $R_2$ intersects $ I_+^\partial (R_1) \cup I_-^\partial(R_1)$.
\end{proof}

\subsection{Time Orientation}
Suppose that $P$ and $Q$ are points in $\tilde{M}$ that are causally related.  Then, the corresponding bulk points, $p$ and $q$ respectively, must either have $p \in J_+(q)$ or 
$q \in J_+(p)$.  But how do we know which?  

There may be a very direct way to answer this question.  Here, however, we give a topological answer.  In appendix \ref{app2} we explain how $\tilde{M}$ be
be made into a topological space.  The basic idea is fairly obvious: two points in $\tilde{M}$ are close to each other if their images under $\mathcal{Q}$ are close.
Because this topology will be consistent with the bulk topology on $\Loc(M)$, we can make use of topological features of the causal structure of the spacetime $\Loc(M)$.

Of particular use is the fact that if $p \in M$, $J_+(p)$ is connected (as is $J_-(p)$).  Because $\Loc(M)$  may be a proper subset of $M$, it is possible that $J_+(p) \cap \Loc(M)$ is not
connected.  Nonetheless, we can consider the connected component of $J_+(p)\cap \Loc(M)$ that contains $p$.  The same construction must be possible in $\tilde{M}$,
but we have to be somewhat more careful.  For $P \in \tilde{M}$, we can consider the set of points $K$ that are causally related to $P$.  This includes $P$ itself.  We can then consider $K \setminus \{P\}$
and look at the two connected components of $K$ that are arbitrarily close to $P$. (There must be exactly two such components because the topology on $\tilde{M}$ needs to be consistent with
that of $\Loc(M)$.) We label these two components $\tilde{J}_\pm(P)$ with the understanding that we have yet to determine which component deserves a plus sign and which deserves a minus sign.

Suppose we arbitrarily choose which of the two regions is to be called $\tilde{J}_+(P_0)$ for one particular point $P_0$.  In all but the most pathological of connected spacetimes, this
fixes the time orientation for every other point in the spacetime.  For example, suppose that $P_1$ is another point in $\tilde{M}$ and  that  
$P_1\in \tilde{J}_+(P_0)$. (Here we are making use of our arbitrary decision about $\tilde{J}_+(P_0)$.)  Then, we must assign the orientation at $P_1$ so that  $P_0 \in \tilde{J}_-(P_1)$.
But now, if  we find another point $P_2 \in \tilde{J}_-(P_1)$, we must have that $\tilde{J}_+(P_2)$ contains $P_1$.  Continuing in this way, we can expect to be able to fix the time orientation for every point in $\tilde{M}$ as long
as it is connected.  This process is depicted in figure \ref{fig_orientation}

\begin{figure}
\centering
\includegraphics[width=8cm]{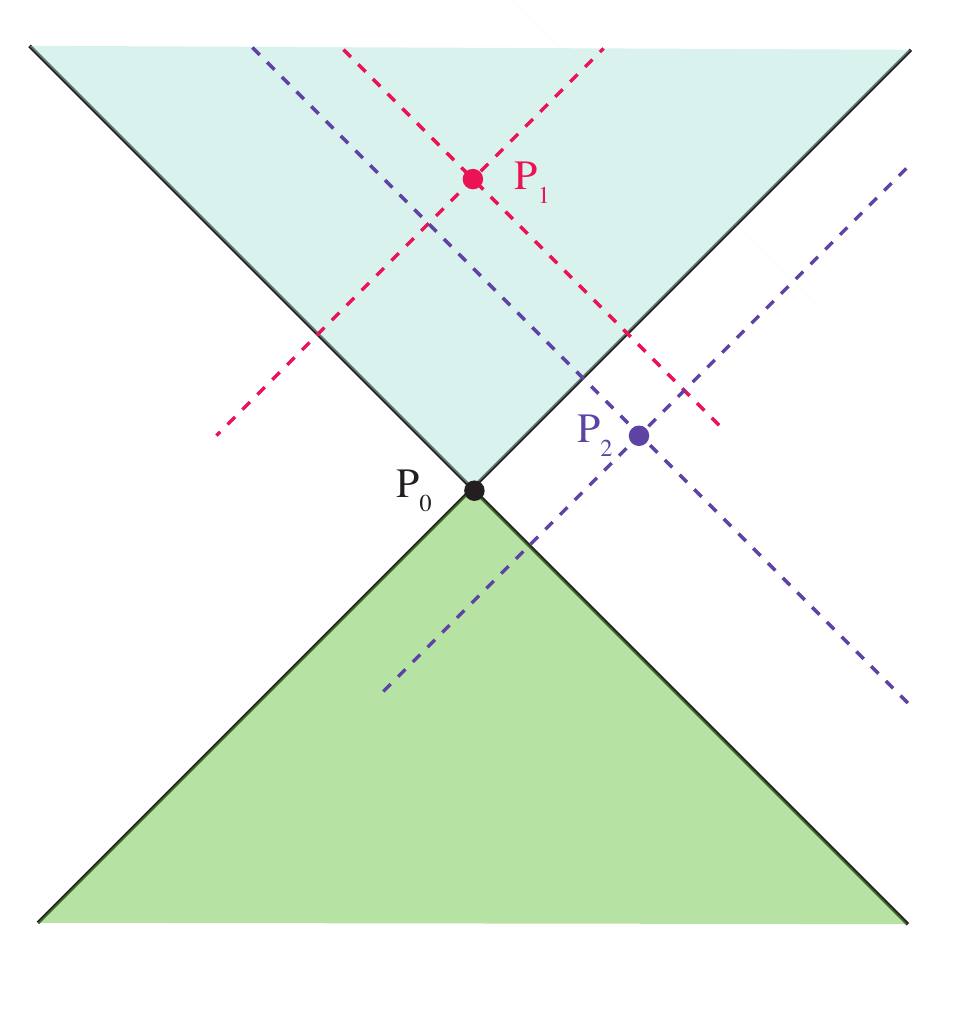}
\caption{
If the definition of the future and past of a point $P_0 \in \tilde{M}$ is chosen, there is an immediate constraint on the time orientation at other points in $\tilde{M}$.  In this figure, the orientation at $P_0$ also
fixes the orientation at $P_1$ and $P_2$.
}
\label{fig_orientation}
\end{figure}

But what about the overall time orientation?  That is, how do we decide on $\tilde{J}_+(P_0)$ in our example above?  This can be done by beginning with a point in $\tilde{M}$ that corresponds
to local boundary operators at some boundary point.  On $\partial M$, we already have a notion of future and past.  Thus, if we take $P_0$ to be an equivalence class consisting only of local boundary operators at a point
$x \in \partial M$, we can decide upon $\tilde{J}_+(P_0)$ by requiring that if $P_1$ is another class of local boundary operators lying at a point $y$ then  $P_1 \in \tilde{J}_+(P_0)$   only if
 $y \in J_+^\partial(x)$.

 We have now succeeded in defining a causal structure on $\tilde{M}$ that must be consistent with that on $\Loc(M)$.  As a consequence, we have reconstructed the metric
 in $\Loc(M)$ up to an undetermined conformal factor.

 \subsection{Comparison with Light-Cone Cut Reconstruction}

There is a compelling connection between the bulk reconstruction developed here and a recent approach to bulk reconstruction involving light-cone cuts due to
 Engelhardt and Horowitz  \cite{Engelhardt:2016wgb}.  Cut reconstruction is a new area of research \cite{Engelhardt:2016, Engelhardt:2016vdk,Engelhardt:2016crc,Engelhardt:2017lfb},
 and remarkably, a number of the ideas involving cuts appear to have analogs in superficial locality reconstruction.  We now detail the similarities and differences between the two approaches.

 
 \begin{itemize}
 
 \item \emph{Large $N$}:  Both cut reconstruction and reconstruction with superficial locality require in their current forms that the classical limit be taken. Light cone cuts are associated with singularities
 of correlation functions of local boundary operators that only resolve in the large $N$ limit.
 These singularities in the boundary theory are at first mysterious but have a simple explanation if one knows about the dual bulk: if there is a bulk point $p$ in the causal wedge of the boundary, then cut singularities are 
 singularities of boundary $n$-point functions $<O(x_1) \ldots O(x_n)>$ that can occur when the boundary points lie on
 the \emph{future and past cuts} of $p$: $C_\pm(p) =( \partial J_\pm(p)) \cap \partial M$.
  These singularities are generally known as bulk-point singularities and have been considered in several contexts prior to that of cuts \cite{Gary:2009ae,Maldacena:2015}.  In particular, \cite{Maldacena:2015} provided
  an example showing that such singularities are not expected to arise without sending $N$ to infinity.  This is consistent with the fact that there should not be any notion of a local bulk scattering point when $N$ is finite.
  Similarly, our consideration of superficially local operators and their equivalence classes is certainly only expected to reproduce local bulk physics in the large $N$ limit.  At finite $N$ there are no local (gauge-invariant)
 observables in the quantum gravity \cite{Donnelly:2015hta, Donnelly:2016}, so it is not clear why one would even seek to study any notion of exactly local bulk operators in this case.  It is, of course,
  interesting to contemplate whether or not either of these approaches suggests new ways to think about approximate locality at finite but large $N$.
   
 \item \emph{Specification of a state}: Cut reconstruction, in its original form, presupposes that we are given a particular quantum state $\psi$ in the CFT Hilbert space and that we are told that $\psi$ is dual to
 some $\emph{unknown}$ bulk geometry.\footnote{In \cite{Engelhardt:2017lfb}, the theory of cuts was put into a framework that did not strictly rely on the presumption of the existence of a bulk, but where an extra dimension can be seen
 to emerge in appropriate cases.} 
 The task is then to study correlation functions in that state (which can be done using the boundary theory only) to determine aspects of the bulk interpretation
 of $\psi$ (like the bulk geometry).
 Similarly, throughout this paper we have assumed that we are given a code subspace $G$ and that we are told that $G$ has the bulk interpretation of being the Hilbert space of
 a quantum field theory on some unknown spacetime background.  We then consider various operators acting on $G$ and ask which of them are superficially local (which can be done using the boundary theory only).

 \item \emph{Identification of points with a boundary object}: The next step in cut reconstruction is to make an identification between the set of light-cone cuts and the set of points in the causal
 wedge of the boundary.  On the other hand, here we identify points in $\Loc(M)$ with equivalence classes of superficially local operators. 
 
 \item \emph{Reconstruction of the Conformal Metric}: It is possible to assign a causal structure to the set of cuts.  This causal structure is consistent with the causal
 structure in the set of bulk points corresponding to the cuts (with some caveats that can be addressed).  As a result, the set of cuts provides a reconstruction of the metric in the causal wedge of the boundary,
 $\mathrm{CW}(\partial M)$, up to a conformal factor.  Similarly, we are able to identify a causal structure on $\tilde{M}$, the set of classes of superficially local operators, 
 and we  therefore obtain a reconstruction of the conformal class of the metric in the bulk region $\Loc(M)$.  It is known that in some cases, $\Loc(M)$ extends further 
 into the bulk than $\mathrm{CW}(\partial M)$ does: in the case of a dynamical black hole, $\Loc(M)$ can intersect the black hole interior.  
 We do not know whether or not it is always the case that $\mathrm{CW}(\partial M) \subseteq \Loc(M)$.
 
 \item \emph{Local operators and the connection between the two methods}: The premise of our approach was to solve a  different problem from bulk reconstruction.  Superficial locality
 provides a way to identify the operators on a code subspace $G$ that are dual to local bulk operators.  Identification of bulk local operators has not yet been a goal of light-cone cut reconstruction, but
 it is a promising direction.   In fact, such considerations suggest a way to directly relate cut reconstruction to our program.  Consider a point $P\in \tilde{M}$ and also consider a light-cone cut $C_\pm$
 associated with singularities in correlation functions computed in a state $\psi \in G$.
 We would like to know how to tell if the bulk point associated with $P$ is the same as the bulk point associated with $C_\pm$ (clearly this is only plausible for bulk points in the intersection 
 of $\mathrm{CW}(\partial M)$ and $\Loc(M)$.

We suggest  the following  approach to this problem.  Consider a superficially local operator $\phi \in P$ and take a collection of boundary points $x_1, \ldots , x_n$ close to points in $C$.  Now, consider two different correlation functions:
\begin{align*}
 F_n(x_1, \ldots, x_n) & = \bra{\psi} O(x_1) \ldots O(x_n) \ket{\psi} \\
 G_n(x_1, \ldots, x_n) & = \bra{\psi}  \phi \:  O(x_1) \ldots O(x_n) \ket{\psi}.
\end{align*}
If $\phi$ is indeed a local operator at the vertex of the cut $C$, then a signature of that property will be encoded in the relationships between $F_n$ and $G_m$ for various values of $n$ and $m$.  We do not pursue this idea
further in the present work.

 
 \end{itemize}

\section{Discussion}
\label{sec_5}


Relying only on subregion duality between the boundary and bulk spacetimes, our construction addresses the following question. Given a CFT and a code subspace dual to an unknown geometry, can we tell if some operator is dual to a bulk local operator? To answer this question, we exploit the curious feature that numerous distinct boundary regions can reconstruct a local bulk operator. Once we identify the set of local bulk operators in the localizable region,
the relations among those operators reveal bulk causal structure.

Furthermore, because the program focuses on entanglement wedges, as opposed to the extremal surfaces themselves, the operators we identify can lie behind horizons and within entanglement shadows in many examples. As expected, however, there are still regions for which our procedure fails to completely describe locality (these regions are often behind horizons). If we assume bulk locality still holds even within these regions, its encoding in the CFT is different than that of operators in the localizable region. 

Subregion duality is a common property of holography. The holographic entanglement entropy prescription \cite{Ryu:2006bv, Hubeny:2007xt} and the fact that entanglement wedge reconstruction is possible \cite{DHW}, lead us to the conclusion that quantum error correction is a feature of any theory with a holographic description. This is an extra constraint on holographic CFTs, which must encode information in a way consistent with bulk reconstruction, and can be seen as a requirement of CFTs having a bulk dual. 

\begin{center}
\textbf{Remaining Considerations}
\end{center}
\emph{Finite $N$}:  While we have addressed how locality, for the portion of the bulk in the localizable region, emerges from quantum error correction, there are still gaps that need to be understood. To what extent does locality fail at finite $N$? Gravitational effects prohibit the existence of  local bulk observables. However, the quantum error correcting properties of subregion duality hold beyond leading order and it may therefore be elucidating to consider
an approximate form of our approach at finite $N$.  This may shed light on the subtleties of the large $N$ limit and the relationship between exact quantum gravity and the infinite $N$ theory.

 \emph{The conformal factor}: While there is no obvious way to reconstruct the conformal factor on $\tilde{M}$, we can argue that more information than just the causal structure is available to us.
 Consider a point $P \in \tilde{M}$ with the special property that for some boundary region $R \in \mathcal{R}$, both $R$ and its complement $\bar{R}$ lie in $\mathcal{Q}(P)$.
 The only geometrical interpretation of this scenario is that operators in $P$ correspond to a point on the HRT surface $\ext R$.  This means that in addition to the conformal metric
 on $\tilde{M}$ we also know the minimal area anchored extremal surfaces as well as the (regulated) areas of those surfaces, determined by the von Neumann entropies of
 corresponding boundary regions \cite{Ryu:2006bv, Hubeny:2007xt,  Lewkowycz:2013nqa, Dong:2016hjy}.  Noting that stationary surfaces  and their areas
 are not invariant under conformal transformations, the conformal factor on the metric is significantly constrained.  We leave further investigation in this direction to future work.

\vspace{.5cm}

\noindent \emph{Note}: during the final stages of this project, \cite{Kabat:2017mun} appeared on the arXiv which discusses related ideas.

\vspace{.5cm}
\noindent {\bf Acknowledgments} 

\noindent We thank Netta Engelhardt, Sebastian Fischetti, Steve Giddings, Gary Horowitz, Stefan Leichenauer,  and Aitor Lewkowycz for useful discussions.
F.S. was supported by the DOE NNSA Stewardship Science Graduate Fellowship.
 S.J.W. was supported by the Department of Energy, Office of Science, Office of High Energy Physics, under contract No.
 DE-SC0011702, and by Foundational Questions Institute
grant FQXi-RFP-1507. 
\appendix

\section{The large $N$ limit}
\label{app1}

In discussing  the main concepts in the text, we have assumed that local bulk operators exist, hoping to present our construction in an intuitive fashion. However,  exact bulk locality only exists when $N = \infty$, and gravitational effects are turned off. When $N$ is large but finite, gravitational effects demands that any gauge invariant bulk operator will be nonlocal in some way \cite{Donnelly:2015hta, Donnelly:2016}. Nevertheless, bulk effective field theory still makes sense within the code subspace of holographic CFTs. This is possible because nonlocal effects become small, since they come with some positive power of the gravitational coupling. The suppression in $N$ allows us to discuss local bulk fields (perhaps smeared over a region $\sim l_p$) and perturbatively add nonlocal effects (by appropriately dressing the fields for example), so long as we work in the appropriate code subspace.

Here, we explain how the constructions in the main text can be made  precise by appropriately applying the large $N$ limit to decouple nonlocalities due to gravity.  Consider a CFT satisfying the appropriate requirements for having a bulk dual (see e.g. \cite{Heemskerk:2009}). The theory has some parameter, $\epsilon(N)$, which corresponds to the gravitational coupling in the bulk and taking $\epsilon \rightarrow 0$ means turning off gravitational effects (i.e. sending $N\rightarrow \infty$). Different values of $\epsilon$ correspond to different boundary theories (with different central charges) with an associated Hilbert space ${H}_\epsilon$.

For $\epsilon \neq 0$, no gauge-invariant operator $\phi_\epsilon$, restricted to the appropriate code subspace ${G}_\epsilon$, will be local in the bulk.  However, as we decrease $\epsilon$, the strength of nonlocal gravitational effects decreases, and some operators and some operators in the CFT will start to resemble what one expects for local operators in semiclassical field theory; intuitively these would be the operators that would limit to local fields in the $\epsilon = 0$ limit. For example, if we think about semiclassical fields that are gravitationally dressed, the gravitational coupling suppresses the nonlocal dressing. 

Consider now a family of operators, $ \{\phi_\epsilon\}_{\epsilon>0}$, with $\phi_\epsilon$ acting on the code subspace ${G}_\epsilon$ for all $\epsilon >0$.
\footnote{
Decreasing $\epsilon$ decreases the strength of gravitational backreaction. In order to keep any nontrivial background fixed while changing the value of $\epsilon$, we must separate ``background matter'' from excitations. 
As we send $\epsilon \rightarrow 0$, the stress tensor for the background matter must be rescaled appropriately to maintain a nontrivial background.  This emphasizes the subtlety in the definition of ${G}_\epsilon$
}

\begin{definition}
Let $R \in \mathcal{R}$ be a boundary region and let $\bar{R}$ be a complement of $R$.  We say that a family $ \{\phi_\epsilon\}_{\epsilon>0}$ is \emph{reconstructable} in $R$ if for any family of operators  $\{O_\epsilon^ {\bar{R}} \}_{\epsilon>0}$ in the algebra of $\bar{R}$ for $H_\epsilon$ and for any family of states $\{\psi_\epsilon\}_{\epsilon>0}$ with $\psi_\epsilon \in G_\epsilon$,
\be
\lim_{\epsilon \rightarrow 0} \bra{\psi_\epsilon} [\phi_\epsilon, O_{\epsilon}^ {\bar{R}}] \ket{\psi_\epsilon}= 0
\ee
\end{definition}

As reviewed in section \ref{sec_2}, this implies that, when $\epsilon$ is very small, there is some operator $O_{\epsilon}^{R}$ in the algebra of $R$, whose action on $G_\epsilon$ is that of the operator $\phi_\epsilon$ (up to corrections in $\epsilon$).

Note that most of these families of operators will not limit to a semiclassical local bulk field.
The ``limit'' might be a smeared operator in $\EW(R)$ or the family of operators could oscillate forever within $\EW(R)$ and never converge in any sense.
 However, some special class of such families do limit to  local operators. 

In order to test whether or not a collection of operators approaches a local field as $\epsilon$ becomes small, we introduce a  generalization of the procedure in the text.  The idea is to make the fundamental object of
study the collection of $\epsilon$-dependent families of operators as opposed to the set of operators on a fixed code subspace.  Following the framework from  section \ref{sec_3}, we  introduce a
 map $\mathcal{Q}$ that acts on families of operators as follows: 
 
 \be
\mathcal{Q}( \{\phi_\epsilon\}_{\epsilon>0} )= \{R \  \big|  \{\phi_\epsilon\}_{\epsilon>0} \ \textrm{is reconstructible in} \ R \}
\label{Qn}
\ee

For some of these sequences, the set $\mathcal{Q}(\{\phi_\epsilon\}_{\epsilon>0})$  will be the result expected for a field localized to a point in the bulk. If this is the case, we can think of $\{\phi_\epsilon\}_{\epsilon>0}$ as a set of operators whose bulk interpretation is a semiclassical field (built on a background associated with a code subspace) whose nonlocal gravitational effects disappears as $\epsilon \rightarrow 0$. For such sequences of operators, taking the $\epsilon \rightarrow 0$ limit is can be thought of as ``undressing'' $\phi$ by consistently tuning down gravitational effects while keeping the background fixed.

We can use this new definition of $\mathcal{Q}$ to define equivalence classes of families of operators and then the notion of superficial locality\footnote{Note that the definition of superficial locality works very nicely with our new definition of $\mathcal{Q}$.
If it happened, for example,  that $\{\phi_\epsilon\}_{\epsilon>0}$ were a family of operators that oscillates from place to place as $\epsilon \rightarrow 0$, then we can be sure that this family would not be superficially local unless
it were to oscillate within a clump.} 
 exactly as we do in section \ref{sec_3}.   All of the  developments in the main text can be done in this formalism.

\section{Topology of $\tilde{M} $}
\label{app2}

In this appendix we explain how a topology on $\tilde{M}$ can be constructed using only the boundary theory.  We make no assumptions here
about whether or not clumps are present. Despite appearances, the purpose of this construction is not so much to demonstrate mathematical rigor as it is to provide motivation
for the statement that  $\tilde{M}$, an object defined in the boundary theory, can be regarded (in the absence of clumps) as a ``copy'' of $\Loc(M)$, a region of spacetime that certainly has a nice topological structure.

The boundary theory is taken to be on a flat space which, after conformal compactification, is a cylinder.  (The case where there are multiple disconnected
boundaries is a straightforward generalization of the construction below.)  A spatial region $R\in \mathcal{R}$
is thus bounded so its boundary, $\partial R$, is compact.  Choose some global coordinate system on this flat spacetime (that is, fix a conformal frame), and define a Euclidean
metric $d$ between two points via geodesic (Euclidean) distance.  We can now give a metric on $\mathcal{R}$ denoted by $D$, by defining $D(R,R^\prime)$ as the 
Hausdorff distance between $\partial R$ and $\partial R^\prime$.%
\footnote{Given a metric space $(S,d)$, the Hausdorff distance is a metric-like function that can be defined in terms of $d$ to measure the distance between
two subsets of $S$ in a reasonable fashion.  The Hausdorff distance is a legitimate metric on the collection of nonempty compact subsets of $S$ so our definition of $D$ provides a metric on the
subset of $\mathcal{R}$ where $\partial R \neq \emptyset$
because $\partial R$ is always compact. }
This definition of distance is problematic in the case where $\partial R = \emptyset$.  However,  if $\partial R_1, \partial R_2= \emptyset$ and $\partial R_3 \neq \emptyset$, we simply define
$D(R_1,R_2) = 0$ and $D(R_1, R_3) = \infty$. 

Given $\epsilon >0$, let $B_\epsilon(R)$ be the subset of $\mathcal{R}$ consisting of regions $R^\prime$ with $D(R,R^\prime) < \epsilon$.  A topology on $\tilde{M}$ can now be obtained by taking $P \in \tilde{M}$
and defining $U_\epsilon(P)$ as the set of points $P^\prime \in \tilde{M}$ such that for every $R \in \mathcal{Q}(P)$, there exists $R^\prime \in \mathcal{Q}(P^\prime) \cap B_\epsilon(R)$.  The collection of
sets $\{U_\epsilon(P) \: \big| \: \epsilon>0, P \in \tilde{M} \}$ forms a topological base from which a topology can be defined.


\end{document}